\documentclass{tlp}

\usepackage[utf8]{inputenc}
\usepackage{graphicx}

\usepackage[fleqn]{amsmath} 
\usepackage{algorithm}
\usepackage[noend]{algpseudocode}

\usepackage{alltt}
\usepackage{tikz}
\usepackage{etex}
\usepackage[etex=true,export]{adjustbox}
\usepackage{framed} %

\usepackage{amssymb}
\usepackage{upquote}
\usepackage{fancyvrb} %

\usepackage{float}
\usepackage{wrapfig}

\usepackage{hyperref}
\usepackage{listings}
\usepackage{courier}
\usepackage{paralist}

\newcommand{\ignore}[1]{}

\newcommand{\pcode}[2][\codesize]{
    \fbox{
    \begin{minipage}{0.45\linewidth}
    #1
    \begin{tabbing}
    xx \= xx \= xx \= xx \= xx \= xx \= xx \= xx \= xx \= xx \= \kill
    #2
    \end{tabbing}
    \end{minipage}
    }
  }

\title{Tree dimension in verification of constrained Horn clauses}

\author[Kafle, Gallagher and Ganty]
{Bishoksan Kafle\\
The University of Melbourne\\
E-mail: bishoksank@unimelb.edu.au\\
\and
John P. Gallagher\\
Roskilde University and IMDEA Software Institute\\
E-mail: jpg@ruc.dk\\
\and 
Pierre Ganty\\
IMDEA Software Institute\\
Email: pierre.ganty@imdea.org
}

\newcommand{\rahft}{{\sf \textsc{Rahft}}}

\newcommand{\false}{\mathtt{false}}
\newcommand{\falseit}{\mathit{false}}
\newcommand{\trueit}{\mathit{true}}

\newcommand{\lfp}{{\sf lfp}}

\newcommand{\abst}{\mathsf{abstract}}

\newcommand{\vars}{\mathsf{vars}}

\newcommand{\renameunfold}{\mathsf{pe\_cls}}
\newcommand{\rep}{\mathsf{rep}}

\newcommand{\constr}{\mathsf{constr}}

\newcommand{\SAT}{\mathsf{SAT}}

\newcommand{\unfold}{{\sf pe\_step}}

\newcommand{\AND}{{\sf AND}}

\newcommand{\Safe}{{\sf Safe}}
\newcommand{\Subst}{{\sf subst}}

\newcommand{\safe}{{\sf safe}}
\newcommand{\unsafe}{{\sf unsafe}}
\newcommand{\stunknown}{{\sf unknown}}
\newcommand{\solution}{{\sf solution}}
\newcommand{\counterexample}{{\sf counterexample}}
\newcommand{\status}{{\sf status}}

\newcommand{\atmost}[1]{\le\!#1}
\newcommand{\exactly}[1]{=\! #1}
\newcommand{\exceeds}[1]{>\! #1}
\newcommand{\anydim}[1]{\ge\! 0}

\def\ll{[\![}

     {\end{tabular}\end{tt}\end{small}}
\def\anno#1{{\ooalign{\hfil\raise.07ex\hbox{\small{\rm #1}}\hfil%
        \crcr\mathhexbox20D}}}

\newcommand{\tuplevar}[1]{\mathbf{#1}}

\newcommand{\theory}{\mathbb{T}}

\newtheorem{definition}{Definition}
\newtheorem{example}{Example}

\newtheorem{proposition}{Proposition}

\begin{document}
\maketitle

\begin{abstract}
In this paper we show how the notion of tree dimension can be used in the verification of constrained Horn clauses (CHCs). The dimension of a tree is a numerical measure of its branching complexity and the concept here applies to Horn clause derivation trees. Derivation trees of dimension zero correspond to derivations using linear CHCs, while trees of higher dimension arise from derivations using non-linear CHCs.   We show how to instrument CHCs predicates with an extra argument for the dimension, allowing a CHC verifier to reason about bounds on the dimension of derivations. Given a set of CHCs $P$, we define a transformation of $P$ yielding a \emph{dimension bounded} set of CHCs $P^{\atmost{k}}$.  The set of derivations for $P^{\atmost{k}}$ consists of the 
derivations for $P$ that have dimension at most $k$.  We also show how to construct a set of clauses denoted $P^{\exceeds{k}}$ whose derivations have dimension exceeding $k$.  We then present algorithms using these constructions to decompose a CHC verification problem.  
One variation of this decomposition considers derivations of successively increasing dimension.  
The paper includes descriptions of implementations and experimental results.
\end{abstract}

%
%
\section{Introduction}
The dimension of a tree, also known as the Horton-Strahler number of a tree\footnote{https://en.wikipedia.org/wiki/Strahler\_number} is a numerical measure of a tree's branching complexity. The concept was originally applied to analyse flows in rivers and their tributaries and to other naturally occurring tree structures \cite{DBLP:conf/lata/EsparzaLS14}. Recently it has found several applications in program analysis and verification \cite{DBLP:journals/jacm/EsparzaKL10,DBLP:conf/popl/RepsTP16}.  In this paper we apply the notion of tree dimension in the verification of CHCs, where the trees whose dimension we consider are derivation trees.  Derivation trees of dimension zero correspond to derivations using linear CHCs, while trees of higher dimension arise from derivations using non-linear  CHCs. 

The verification of a property of a set of CHCs often involves implicitly the set of all derivation trees for that set.  For example, a safety property is typically formalised as the consistency of a set of clauses, which amounts to establishing the absence of a derivation of a contradiction and requires the consideration of all derivations for the given clauses. CHCs provide a convenient representation for the statement of invariant properties of various systems including imperative programs \cite{DBLP:conf/pldi/GrebenshchikovLPR12}, which is again usually formalised as the absence of a derivation of some statement representing the violation of the invariant. An automated tool for finding such derivations might benefit from a divide-and-conquer strategy, decomposing the set of all derivations into smaller more manageable sets.

Tree dimension provides one such approach to decompose verification problems that involve the set of all derivations. Given a set of CHCs $P$ and a dimension $k\ge 0$, we define a transformation yielding a set of CHCs $P^{\atmost{k}}$ whose derivations have dimension of at most $k$. We can also obtain the complementary set of clauses (called $P^{\exceeds{k}}$) whose derivation trees have dimension at least $k+1$. Each such set of clauses ($P^{\atmost{k}}$ and $P^{\exceeds{k}}$) represents an under-approximation of the original set $P$ in the sense that they give rise to a subset of $P$'s derivations. 

Why might decomposition by dimension be useful? Firstly, the overall verification problem is reduced into simpler, but still non-trivial parts, each with an infinite number of derivations. By contrast, if one of the parts were finite, say the set of derivations of bounded depth, then the complementary part would arguably be no simpler than the original. Secondly, the particular properties of bounded dimension can be exploited. Any dimension-bounded set of clauses $P^{\atmost{k}}$ can be linearised, while preserving key semantic properties including consistency \cite{DBLP:journals/corr/KafleGG16}. This allows the use of tools designed and optimised for linear clauses. 

We also show how to reason directly about the dimension of derivations using any CHC verification system, by instrumenting the clauses, adding an extra argument to each predicate representing the dimension.

In Section \ref{prelim} we introduce the technical background of the paper. We review the notion of tree dimension and introduce the syntax and semantics of CHCs. We relate the concept of \emph{tree dimension} to CHCs derived from imperative programs in  Section \ref{ch:dimension_bounded_chc}; and present  a method for instrumenting CHCs predicates with an extra argument for the dimension and verify dimension related properties using the standard CHCs solvers. In Section \ref{sec:pe} we present partial evaluation algorithms to construct two versions of dimension-bounded clauses constructed from a given set of CHCs: one whose derivations are bounded in dimension from above and one whose derivations are bounded from below. 
The dimension-bounded sets of clauses are exploited by verification algorithms presented in Section \ref{sec:verification_algorithm}.
Section \ref{experiments} contains a description of a prototype implementation and discusses the results obtained. Section \ref{rel} presents a discussion of related work as well as the role of dimension in using CHCs for safety verification of imperative programs. Finally, Section \ref{concl} concludes.
 
%
%
%
%
%
%
%
%
%
%
%
%
%
%
%
%
%
%
%
%

 %
%
\section{Preliminaries and formal background}\label{prelim}

\label{sec:tree_dimension}

A labelled tree $c(t_1,\ldots,t_k)$ ($k \ge 0$) is a tree whose nodes are labelled by identifiers, where $c$ is the label of the root and $t_1,\ldots,t_k$ are labelled trees, the children of the root.  
In this paper, all trees we consider are finite.

The dimension of a tree is a measure of its non-linearity; for example a linear tree (whose nodes have at most one child) has dimension zero while a complete binary tree has dimension equal to its height.  Formally, the dimension of a tree is defined as follows.

\begin{definition}[Tree dimension adapted from \citeN{DBLP:conf/stacs/EsparzaKL07}]\label{treedim}
  Given a labelled tree $t=c(t_1,\ldots,t_k)$, the tree dimension of $t$ represented as \(\mathit{dim}(t)\) is defined as follows:
  \[
  \mathit{dim}(c(t_1,\ldots,t_k))= \begin{cases}
  0 & \text{if } k=0 \\ 
  \mathit{dim}(t_i)&\text{if } k>0 \land \vert\{ i \mid \forall j\colon \mathit{dim}(t_j)\leq\mathit{dim}(t_i) \}\vert = 1 \\
\mathit{dim}(t_i) +1 &\text{if } k>0 \land \vert\{ i \mid \forall j\colon \mathit{dim}(t_j)\leq\mathit{dim}(t_i) \}\vert > 1 
\end{cases}
  \]
\end{definition}

 Figure~\ref{fig:treedimension} shows a  labelled tree $t=c_3(c_2(c_2(c_1,c_1), c_1))$ (each $c_i$ is a node label) in graphical form and the dimension of each of its subtrees. The dimension of the root node (1 in this case) is the dimension of the tree. 

\begin{figure}[h!]
  \centering

    \includegraphics[width=0.80 \textwidth]{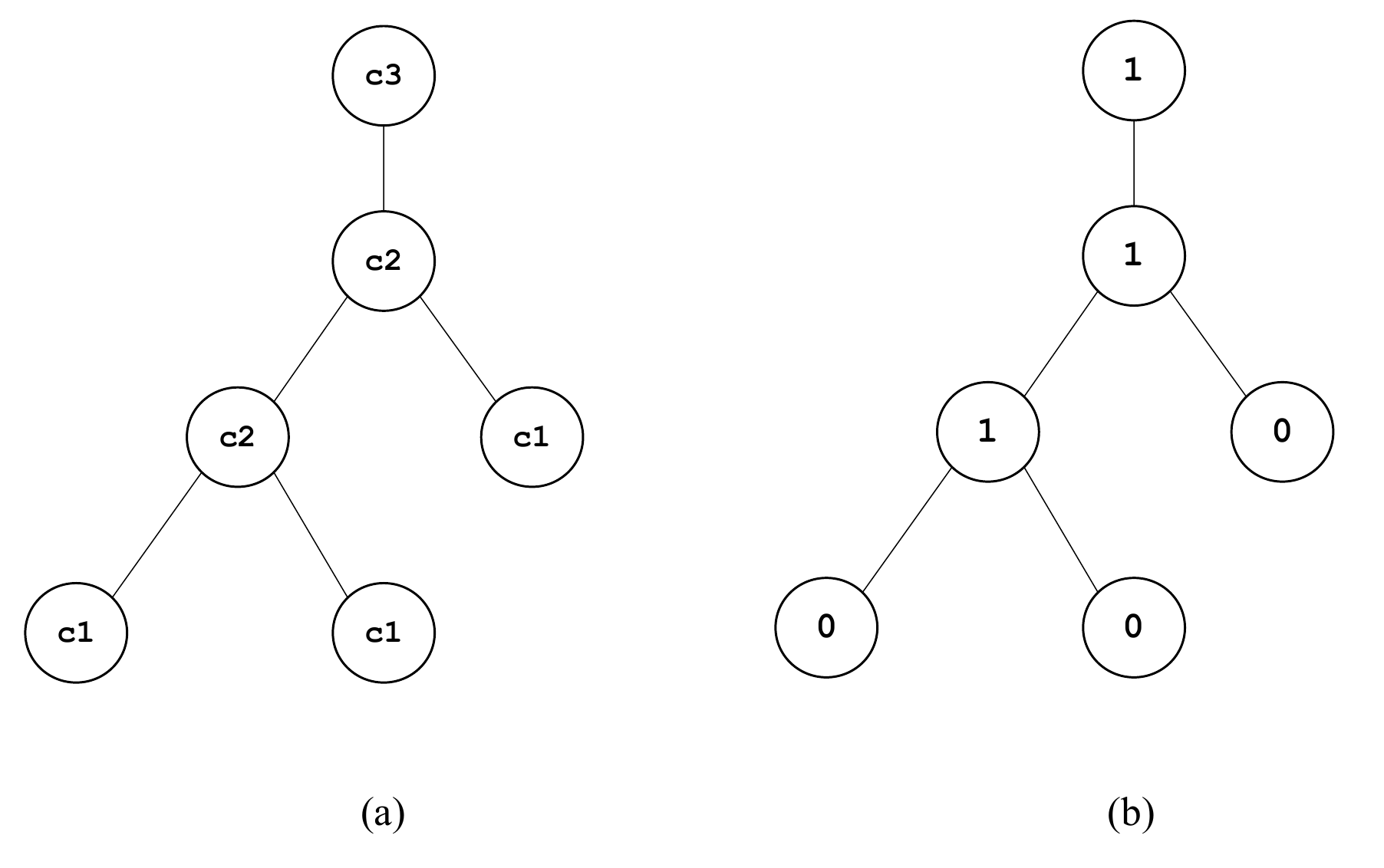}

    \caption{(a) a labelled tree $c_3(c_2(c_2(c_1,c_1), c_1))$ and (b) the dimension of each subtree.}
    \label{fig:treedimension}
\end{figure}
 %
%
%
\label{sec:chcs}

\begin{figure}[t]
  \begin{center}
  \begin{tabular}{ll}
    \pcode[\small]{

 $F_n=
\begin{cases}
n \mathrm{~~~if~ }n=0 ~\mathrm{or}~ 1 \\
F_{n-1} + F_{n-2} \mathrm{~~if~ }n>1
\end{cases}
$
} &
\begin{lstlisting}[linewidth=6.3cm]
c1. fib(A,B):- A>=0, A=<1, B=A.
c2. fib(A,B):- A>1, A2=A-2,  
            A1=A-1, fib(A2,B2),   
           fib(A1,B1), B=B1+B2. 
c3. false:- A>5, fib(A,B), B<A.    
\end{lstlisting} 
\end{tabular}
  \end{center}

\caption{Fibonacci function (left), its encoding as CHCs  and a property \textit{Fib} (right).}
 \label{exprogram}
\end{figure}

A constrained Horn clause (CHC) is a first-order predicate logic formula of the form $\forall \tuplevar{x_0} \ldots \tuplevar{x_k} 
(p_1(\tuplevar{x_1}) \wedge \ldots \wedge p_k(\tuplevar{x_k}) \wedge \phi \rightarrow
p_0(\tuplevar{x_0}))$, where $\phi$ is a finite conjunction of \emph{constraints} with
respect to some constraint theory, $\tuplevar{x_0},\ldots, \tuplevar{x_k}$ are (possibly empty) tuples of \emph{variables}, $p_0,\ldots,p_k$ are \emph{predicate symbols}, $p_0(\tuplevar{x_0})$ is the \emph{head} of the clause and $p_1(\tuplevar{x_1}) \wedge \ldots \wedge p_k(\tuplevar{x_k}) \wedge \phi$ is the \emph{body}.
Following the conventions of Constraint Logic Programming (CLP), such a clause is written as $p_0(\tuplevar{x_0}) \leftarrow \phi , p_1(\tuplevar{x_1}) , \ldots , p_k(\tuplevar{x_k})$. 
An atomic formula, or simply \emph{atom}, is a formula $p(\tuplevar{x})$ where $p$ is a
predicate symbol and $\tuplevar{x}$ a tuple of arguments.  Atoms are sometimes written as $A$, $B$ or $H$, possibly with sub- or superscripts.

A clause is called \emph{non-linear} if it contains more than one atom in the body, otherwise it is called \emph{linear}.  A set of CHCs $P$ is called  linear if $P$ only contains linear clauses, otherwise it is called non-linear.  \emph{Integrity constraints} are a special kind of clause whose head is the predicate  $\mathtt{false}$. A set of constrained Horn clauses can also be regarded as a \emph{constraint logic program}, though in  this paper CHCs are not regarded as executable programs; we are concerned with verifying logical properties of CHCs.

For concrete examples of CHCs we use Prolog syntax and typewriter font, writing the implication $\leftarrow$ as \texttt{:-} and using capital letters for variable names.  The constraints can also be intermixed with the body atoms.
Figure \ref{exprogram} (right) contains an example of a set of constrained Horn clauses, called \textit{Fib}, which encodes the Fibonacci function.
The first two clauses \texttt{c1} and \texttt{c2} define the Fibonacci function and clause \texttt{c3} represents a property of the Fibonacci function expressed as an integrity constraint.  \texttt{c2} is a non-linear clause while  \texttt{c1} and \texttt{c3} are linear.
Each CHC in a given set of CHCs is associated with an identifier, as illustrated in Figure \ref{exprogram}.

\paragraph{CHC semantics.} The semantics of CHCs is obtained using standard concepts from predicate logic semantics.  
An \emph{interpretation} assigns  to each predicate a relation over the domain of the constraint theory $\theory$, whereas constraints have interpretations in the theory itself. 
In particular, the predicate $\false$ is always interpreted as $\falseit$.
An interpretation \emph{satisfies} a set of formulas if each formula in the set evaluates to $\trueit$ in the interpretation in the standard way.
In particular, \emph{a model} of a set of CHCs is an interpretation in which each clause evaluates to $\trueit$. 
A set of CHCs $P$ is \emph{consistent} if and only if it has a model. 
Otherwise it is \emph{inconsistent}. 

In the algorithms developed in Section \ref{sec:verification_algorithm}, we consider only interpretations representable within the constraint theory by a set of \emph{constrained facts} of the form $p(\tuplevar{x}) \leftarrow \phi$ where $\tuplevar{x}$ is a tuple of distinct variables and $\phi$ a  constraint  (free variables of $\phi$ are subset of $\tuplevar{x}$)  in the constraint theory underlying the CHCs.
There is exactly one constrained fact for each predicate \(p\)  in the set of CHCs.
Such a constrained fact defines the interpretation of $p$ as the relation \(\{\tuplevar{x}\theta \mid \tuplevar{x}\theta \text{ is ground, and }\phi\theta \text{ holds in } \theory \} \).
We call such a set of constrained facts a \emph{syntactic interpretation}, and if it is a model, we call it a \emph{syntactic model}.
If a set of CHCs has a syntactic model, then it has a model, but the reverse is not necessarily true.
In particular, a syntactic interpretation satisfies a clause $A_0 \leftarrow   \phi, A_1, \ldots, A_k$ if for constrained facts (with variables suitably renamed) $A_0 \leftarrow \phi_0$, $A_1 \leftarrow \phi_1$, $\ldots$, $A_k \leftarrow \phi_k$ in the interpretation, the formula $\phi \wedge \phi_1 \wedge \ldots \wedge \phi_k \rightarrow \phi_0$ holds in the underlying constraint theory.

In some works e.g. \cite{DBLP:conf/sas/BjornerMR13,McmillanR2013}  a syntactic model is also called a \emph{solution} and we use these terms interchangeably in this paper when the context is clear.  
When modelling safety properties of systems using CHCs, the consistency of a set of CHCs corresponds to \emph{safety} of the system.   
Thus we also refer to CHCs as being \emph{safe} or \emph{unsafe} when they are consistent or inconsistent respectively.

\paragraph{AND-trees and trace trees.} Derivations for CHCs are represented by AND-trees.
The following definitions of derivations and trace trees are adapted from 
\citeN{Gallagher-Lafave-Dagstuhl}. From now on, we assume that each clause has a unique identifier.

\begin{definition}[AND-tree or derivation tree]\label{def:andtree}
An \emph{AND-tree} for a set of CHCs is a tree each of whose nodes is labelled by an atom, a constraint and a clause identifier such that
\begin{enumerate}
\item
each non-leaf node corresponds to a clause (with variables suitably renamed)
$A \leftarrow \phi, A_1,\ldots,A_k$ and is labelled by an atom $A$, constraint $\phi$ and has children labelled by atoms $A_1,\ldots,A_k$;
\item
each leaf node corresponds to a clause $A \leftarrow \phi$ (with variables suitably renamed) and is labelled by
an atom $A$ and constraint $\phi$;
\item each node is labelled with the clause identifier of the clause corresponding to the node.
\end{enumerate}
The phrase “with variables suitably renamed” here and elsewhere in the paper means that variables occurring in the body but not in the head do not occur in the labels of any ancestor node. An example of an AND-tree is shown in Figure \ref{fig:tracetree} (right).
\end{definition}

\begin{figure}[ht!]
  \centering
    \includegraphics[width= \textwidth]{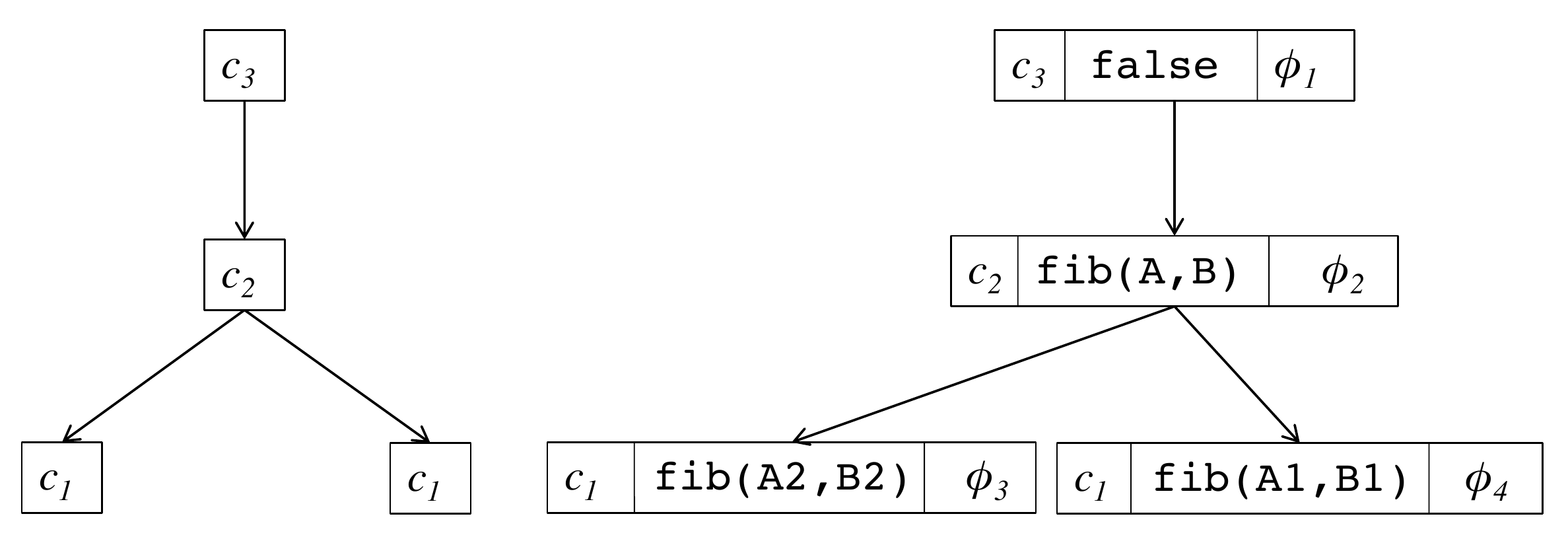}
    \caption{A trace-term $c_3(c_2(c_1,c_1))$ of Fib  (left) and its $\AND$-tree (right), where $\phi_1 \equiv \mathtt{A>5} \wedge \mathtt{B<A}; ~\phi_2 \equiv  \mathtt{A>1} \wedge \mathtt{A2=A-2} \wedge  \mathtt{A1=A-1} \wedge  \mathtt{B=B1+B2};  ~\phi_3 \equiv \mathtt{A2\geq\ 0} \wedge  
\mathtt{A2\leq\ 1} \wedge \mathtt{B2=A2}; ~\phi_4 \equiv \mathtt{A1\geq\ 0} \wedge \mathtt{A1\leq\ 1} \wedge \mathtt{B1=A1}$.}
   \label{fig:tracetree}
\end{figure}

\begin{definition}[$\constr(t)$]\label{constr}
Given an AND-tree $t$, the conjunction of the constraints  in its node labels is represented by $\constr(t)$. $t$  is \emph{feasible} or \emph{successful} if and only if $\constr(t)$  is satisfiable under $\theory$. 
\end{definition}
\begin{definition}\label{provable}
For an atom $p(\tuplevar{x})$ and a set of CHCs \(P\) we write $P \vdash p(\tuplevar{x})$ if there exists a feasible AND-tree with root labelled by $p(\tuplevar{x})$. 
\end{definition}

\begin{definition}\label{counterexample}
A feasible AND-tree with root node labelled by $\mathtt{false}$ is called a \emph{counterexample}.
\end{definition}
The soundness and completeness of derivation trees \cite{JMMS} implies that $P$ is inconsistent if and only if $P \vdash \mathtt{false}$, that is, $P$ has a counterexample. AND-trees in this paper, unless otherwise stated, are counterexamples.

An AND-tree $t$ can be associated with a more abstract structure called a \emph{trace tree}, which is the result of removing all node labels from $t$ apart from the clause identifiers. 
The identifiers can be treated as constructors whose arity is the number of atoms in the clause body of the clause associated with the identifier.  
In this way we can write trace trees as terms (as in Figure \ref{fig:treedimension}(a)).

Thus a trace tree, together with a mapping from clause identifiers to clauses, uniquely defines an AND-tree (up to renaming of variables). Namely, $c(t_1,\ldots,t_k)$ corresponds to the AND-tree whose root is labelled by the atom $A$ and the clause $A \leftarrow   \phi, A_1, \ldots, A_k$ whose identifier is $c$,  and whose children are the AND-trees corresponding to $t_1,\ldots,t_k$ respectively. 

\begin{definition}[Dimension of a CHC derivation]\label{def:deriv-dim}
The dimension of a derivation for a set of CHCs is the tree dimension of the AND-tree (or associated trace tree) for the derivation.
\end{definition}

It is clear from these definitions that the dimension of derivations is closely related to the syntactic structure of CHCs.  For instance, a set of linear clauses can give rise only to derivations of dimension zero, since the corresponding trace trees are linear.


%
%
\section{Tree dimension and CHCs}
\label{ch:dimension_bounded_chc}

\subsection{Programs as CHCs and their dimension}
In this subsection we discuss the notion of tree dimension in relation to CHCs representing imperative programs. CHCs provide a suitable language for expressing the semantics of imperative languages \cite{Peralta-Gallagher-Saglam-SAS98,DBLP:conf/birthday/BjornerGMR15,DBLP:conf/tacas/GrebenshchikovGLPR12}, enabling the use of CHC tools for verification of properties of imperative programs. 
The clauses resulting from the translation may give rise to derivations of different dimension, depending on the style of semantic specification underlying the translation. For example, procedures call can be encoded as linear (consider inline) or non-linear CHCs  giving rise to different dimensions.

\paragraph{Imperative programs without procedures.} Consider first a language with no procedures. Let $S$  be an imperative statement such as an assignment, conditional or loop and let a configuration $\langle S, \sigma \rangle$ stand for statement $S$ executing in state $\sigma$. In \emph{structural operational semantics} \cite{DBLP:books/daglib/0067731} (sometimes called small-step semantics), the meaning of statements is expressed by transitions of the form $\langle S, \sigma \rangle \Rightarrow \langle S', \sigma' \rangle$, which means that executing $S$ in state $\sigma$ yields (in one execution step) the configuration $\langle S', \sigma' \rangle$.  A translation based on small-step semantics then yields a corresponding linear clause $p_S(\sigma) \leftarrow \phi(\sigma,\sigma'), p_{S'}(\sigma')$, 
where $p_S$ and $p_{S'}$ are predicates corresponding to statements $S$ and $S'$ respectively, and $\phi(\sigma,\sigma')$ is a constraint relating the variables in states $\sigma$ and $\sigma'$. (Alternatively, we could choose $p_{S'}(\sigma') \leftarrow \phi(\sigma,\sigma'), p_S(\sigma)$, reversing the direction of the transition, depending on the purpose of the encoding).

By contrast, in \emph{natural semantics} (sometimes called big-step semantics), the meaning of a statement $S$ is expressed by a transition $\langle S, \sigma \rangle \Rightarrow \sigma'$, where this means that the execution of statement $S$ in state $\sigma$ terminates with final state $\sigma'$.  A translation based on big-step semantics yields clauses that break down such a “big step” into smaller steps, using the syntactic structure of the statement.

The difference between the two styles can be clearly seen for the translation of a statement sequence $S_1;S_2$. The small-step semantics would yield linear clauses of the following form, in which the computation of $S_1$ is carried out step by step until $S_1$ terminates, and then $S_2$ is executed. 
\[
\begin{array}{l}
p_{S_1;S_2}(\sigma) \leftarrow \phi_1, p_{S_1';S_2}(\sigma'). \\
\ldots\\
p_{S_1'';S_2}(\sigma) \leftarrow \phi_2, p_{S_2}(\sigma'). \\
p_{S_2}(\sigma) \leftarrow \phi_3, p_{S_2'}(\sigma'). \\
\ldots\\
\end{array}
\]

The clauses resulting from small-step semantics closely correspond to the control-flow graph of the statement, where each clause corresponds to an edge in the graph.

The big-step semantics of $S_1;S_2$ yields a clause of the form:
\[
\begin{array}{l}
p_{S_1;S_2}(\sigma, \sigma'') \leftarrow p_{S_1}(\sigma, \sigma'),p_{S_2}(\sigma',\sigma''). \\
p_{S_1}(\sigma, \sigma') \leftarrow \ldots.\\
p_{S_2}(\sigma', \sigma'') \leftarrow \ldots.\\
\ldots\\
\end{array}
\]

Here the first clause is non-linear, chaining the two big steps corresponding to the execution of $S_1$ and $S_2$ together to make one big step for $S_1;S_2$.

A translation from imperative code to CHCs may mix big- and small-step styles.   
In both styles, a loop results in a recursive predicate 
(that is, one that calls itself directly or indirectly).  
Regarding the dimension of derivations in the two styles, however, it is clear that small-step semantics yields linear clauses and hence zero-dimensional derivations, that is, all derivation trees will be linear.  
Big-step semantics, on the other hand, yields non-linear clauses. However, although the clauses contain recursive predicates for the loops, it can be shown that derivations using the non-linear clauses derived from big-step semantics have bounded dimension, with the bound determined by the level of statement nesting.
Since clauses whose derivations are of bounded dimension can be linearised \cite{DBLP:journals/corr/KafleGG16},  these non-linear clauses can be transformed to linear clauses. It may be asked whether the result is the same as the clauses resulting from the small-step-based translation.  The answer is “not exactly”. While the linearised clauses resulting from big-step semantics would correspond to the same small execution steps, there are more arguments of the predicates than in the clauses resulting from small-step semantics, representing the intermediate states that are created in the clause bodies resulting from big-step semantics.
\begin{example}\label{linearising big-step}
Given the program $P: \ \   \mathtt{x=1; y=2;}$ the small-step encoding gives:
\begin{lstlisting}
s(X,Y):- X1=1, Y1=Y, s1(X1,Y1).
s1(X,Y):- X1=X, Y1=2, s2(X1,Y1).
s2(X,Y):- true.
\end{lstlisting}
The big-step encoding gives:
\begin{lstlisting}
b(X0,Y0,X2,Y2):- b1(X0,Y0,X1,Y1), b2(X1,Y1,X2,Y2).
b1(X0,Y0,X1,Y1):- X1=1, Y1=Y0.
b2(X1,Y1,X2,Y2):- X2=X1, Y2=2.
\end{lstlisting}
This can be straightforwardly linearised to the following, where each predicate represents the remaining computation.  
\begin{lstlisting}
p(X0,Y0,X2,Y2):- p1(X0,Y0,X1,Y1,X2,Y2).
p1(X0,Y0,X1,Y1,X2,Y2):- X1=1, Y1=Y0, p2(X1,Y1,X2,Y2).
p2(X1,Y1,X2,Y2):- X2=X1, Y2=2.
\end{lstlisting}
This is similar to the small-step encoding, but contains more arguments, partly due to the fact that the final state of the small-step encoding is not explicitly returned, but it is returned in the big-step encoding, and partly due to the variables representing intermediate states (for example in the predicate $\mathtt{p1}$).
\end{example}

\paragraph{Imperative programs with procedures.} Turning to a language with procedures, the small-step semantics requires the state to include a stack, whose height is unbounded in the presence of recursive procedures. The \emph{call} and \emph{return} statements respectively push and pop the stack. Thus the clauses, though still linear, are interpreted over a richer domain than that of the program variables themselves. In the big-step semantics no explicit stack is needed; a procedure call is represented, as other statements, with a big-step predicate expressing the relation between the states before and after the call (in effect, the predicate is a procedure summary).  

As regards dimension, clauses resulting from big-step semantics of programs with recursive procedures can give rise to derivations of unbounded dimension due to the presence of  recursive procedures of the form \texttt{proc p() \{\ldots p();\ldots p();\ldots\}}, which yields a non-linear clause of this form.
\[
p(\sigma_0, \sigma_n) \leftarrow \ldots, p(\sigma_1, \sigma_2),\ldots,p(\sigma_3,\sigma_4),\ldots \\
\]
 We note that the clauses due to big-step semantics could still be linearised (in effect a transformation to continuation-passing form in which a stack is introduced) but this transformation is different from the linearisation of bounded-dimension clauses.

 In summary, CHCs representing single imperative procedures with no calls to external procedures are naturally linear, either by direct translation based on small-step semantics (or equivalently, control-flow graphs) or by translating to dimension-bounded clauses using big-step semantics and then linearising using techniques presented in  \citeN{DBLP:journals/corr/KafleGG16} and \citeN{DBLP:journals/tcs/AfratiGT03}.
On the other hand, imperative programs with procedure calls can be given a straightforward translation into CHCs using big-step semantics, but the dimension of derivations in the clauses is not in general bounded. The techniques described in this paper for decomposition based on dimension are hence mostly relevant for verification and analysis of imperative programs with recursive procedures.  Other techniques for obtaining linear clauses from such programs do so at the cost of introducing a stack as a predicate argument.

\subsection{Construction of dimension instrumented set of clauses}
\label{ch:dimension_instrumented}

In some sets of CHCs, the dimension of derivation trees is not bounded, but there is a bound on the dimension of \emph{feasible} derivations.
Figure \ref{mc91program} shows the well known 91-function of McCarthy\footnote{http://en.wikipedia.org/wiki/McCarthy\_91\_function} together with its  constrained Horn clauses representation.

\begin{figure}[t]
  \begin{center}
  \begin{tabular}{ll}
    \pcode[\small]{
 $M(n)=
\begin{cases}
n-10 \mathrm{~~~if~ } n>100 \\
M(M(n+11))  \mathrm{~if~}n\leq100
\end{cases}
$
} &
\begin{lstlisting}[linewidth=5.6cm]
mc91(N,X):- N>100, X=N-10.
mc91(N,X):- N=<100, Y=N+11, mc91(Y,Y2), mc91(Y2,X).   
\end{lstlisting} 
\end{tabular}
  \end{center}

\caption{McCarthy's 91-function and its encoding as CHCs.}
 \label{mc91program}
\end{figure}

Although it is possible to construct derivation trees of arbitrary dimension using the clauses in Figure \ref{mc91program}, the dependencies between the two recursive calls to \texttt{mc91} imply that no \emph{feasible} derivation tree for \texttt{mc91(N,X)} has dimension greater than 2. This is a meta-property of  the set of clauses; however,  as we now show, by instrumenting the clauses with dimensions, such properties can be expressed as safety properties of CHCs.

\begin{definition}[Dimension-instrumented clauses]\label{dim_instr}
Let $P$ be a set of CHCs.  The dimension instrumented set $P_{dim}$ of CHCs  is defined as follows.
\begin{itemize}
\item
For each predicate $p$ of arity $m$ define a predicate $p'$ of arity $m+1$.
\item
For each clause in $P$ of the form
$$p(\tuplevar{x}) \leftarrow \phi, p_1(\tuplevar{x_1}),\ldots,p_n(\tuplevar{x_n})$$
construct a clause 
$$p'(\tuplevar{x},k) \leftarrow \phi, p'_1(\tuplevar{x_1},k_1),\ldots,p'_n(\tuplevar{x_n},k_n), dim([k_1,\ldots,k_n],k)$$
in $P_{dim}$, where $k_1,\ldots,k_n,k$ are fresh variables added as the final argument for their respective predicates, and $dim([k_1,\ldots,k_n],k)$ is defined according to the rules in Definition \ref{treedim}  for determining the dimension $k$ of a tree from the dimensions $k_1,\ldots,k_n$ of the subtrees of the root node. 
\end{itemize}
\end{definition}

\begin{proposition}
Let $P$ be a set of CHCs and $P_{\mathit{dim}}$ be the set of clauses defined from $P$ using Definition \ref{dim_instr}. Then
$P_{\mathit{dim}} \vdash p(\tuplevar{t},k)$ if and only if the atom $p(\tuplevar{t})$ has a derivation of dimension $k$ in $P$.
\end{proposition}

\begin{example}
Figure~\ref{fig:mc91instrumented} lists the dimension-instrumented version of the McCarthy 91-function.
\begin{figure}[ht!]
	\centering
\begin{lstlisting}
mc91(N,X,K):- N>100, X=N-10,  dim([],K). 
mc91(N,X,K):- N=<100, Y=N+11,
               mc91(Y,Y2,K1), mc91(Y2,X,K2), dim([K1,K2],K).
dim([],0).
dim([K1,K2], K3):- K1>=K2+1, K3=K1.
dim([K1,K2], K3):- K2>=K1+1, K3=K2.
dim([K1,K2], K3):- K1=K2, K3=K1+1.
\end{lstlisting}
\caption{Dimension instrumented CHCs for the McCarthy 91-function.}
\label{fig:mc91instrumented}
\end{figure}
\end{example}

\subsection{Verification of dimension properties}
\label{sec:proginstr}
Using the instrumented program we can try to prove information about the dimension, such as upper or lower bounds or other relationships between the dimension and other predicate arguments. 

\begin{example}
	To establish that successful derivations for the atom \texttt{mc91(X,Y)} have dimension at most 2 we add the integrity constraint \texttt{false:- mc91(N,X,K), K>2.} to the dimension-instrumented clauses of Fig~\ref{fig:mc91instrumented}.
The clauses together with the integrity constraint are given to an automatic solver for Horn clauses, e.g. \cite{DBLP:conf/tacas/GrebenshchikovGLPR12,DBLP:conf/cav/KafleGM16}, which is able to prove the safety of the clauses and thus establish the upper bound of 2.
\end{example}

In the next example, we show that the dimension can depend on the values of other predicate arguments.  
\begin{example}
The dimension-instrumented version of the \emph{Fib} clauses is shown in Figure \ref{inprogram}.  The property to be proved is that the dimension of the trees rooted at $\false$ of \emph{Fib} is less than or equal to the  half of \emph{Fib}'s  input value, expressed by the integrity constraint \texttt{false:- fib(A,B,K), 2K-1>=A.} 
Again, this property is established by applying a Horn clause solver to prove the safety of the clauses together with the integrity constraint.
\end{example}

 \begin{figure}
 \centering
\begin{lstlisting}
fib(A,A,K):- A>=0, A=<1, dim([],K). 
fib(A,B,K):- A>1, A2 =A-2, fib(A2,B2,K1),
       A1=A-1, fib(A1,B1,K2), B=B1+B2, dim([K1,K2],K).        
\end{lstlisting}
\caption{Dimension instrumented CHCs for the Fib program.}
 \label{inprogram}
\end{figure}

\begin{example}
We present the well known counting change example taken from \citeN[Chapter 1]{DBLP:books/mit/AbelsonS96}. Figure \ref{counting-change} shows its encoding in CHCs and the Figure \ref{counting-change-dim} shows the dimension-instrumented version of the clauses. The property of interest is to relate the number of different coins (counts) with the  dimension of the derivation of the predicate \texttt{cc}. We can establish that the dimension is at most the number of different coins as expressed by the integrity constraint \texttt{false :- B>=1, K>B, cc(A,B,C,K)}. 
\end{example}

 \begin{figure}
 \centering
\begin{lstlisting}
cc(0,Y,1):- Y>0.
cc(X,_,0):- X<0.
cc(_,Y,0):- Y=<0.
cc(X,Y,Z):- X>0, kinds_of_coins(Y,A), 
               X1=X-A, cc(X1,Y,Z1), 
               Y1=Y-1, cc(X,Y1,Z2), Z=Z1 +Z2.
               
kinds_of_coins(1,1).  kinds_of_coins(2,5).  kinds_of_coins(3,10).
kinds_of_coins(4,25). kinds_of_coins(5,50).
        
\end{lstlisting}
\caption{Counting change example encoded as a set of CHCs.}
 \label{counting-change}
\end{figure}

 \begin{figure}[h!]
 \centering
\begin{lstlisting}
cc(0,Y,1,K):- Y>0, dim([],K).
cc(X,_,0,K):- X<0, dim([],K).
cc(_,Y,0,K):- Y=<0, dim([],K).
cc(X,Y,Z,K):- X>0, kinds_of_coins(Y,A,K0), X1=X-A, 
                  cc(X1,Y,Z1,K1), Y1=Y-1, cc(X,Y1,Z2,K2), 
                  Z=Z1+Z2, dim([K0,K1,K2],K). 
                  
kinds_of_coins(1,1,K):- dim([],K).
kinds_of_coins(2,5,K):- dim([],K).
kinds_of_coins(3,10,K):- dim([],K).
kinds_of_coins(4,25,K):- dim([],K).
kinds_of_coins(5,50,K):- dim([],K).
\end{lstlisting}
\caption{Dimension instrumented CHCs for the Counting change example.}
 \label{counting-change-dim}
\end{figure}

In general, verifying whether all the feasible derivation trees of a predicate in the program has a certain dimension is as challenging as proving any other non-trivial properties of the program. But in some cases the knowledge of  dimension of derivation trees of a program is useful for verifying other program properties.  For instance, using the knowledge that the derivation trees of McCarthy 91-function have dimension at most 2 would allow us to restrict the verification of any program property relating to successful derivations to the derivations in the dimension-bounded program $P^{\atmost{2}}$ (see Section \ref{pe-atmost}) where $P$ is the set of clauses for the McCarthy 91-function.

 
%
\subsection{Derivation of dimension-bounded CHCs by partial evaluation}
\label{sec:pe}

Definition \ref{dim_instr} showed how to construct $P_{\mathit{dim}}$, an ``instrumented" version of a set of CHCs $P$, such 
that $P_{\mathit{dim}} \vdash p(\tuplevar{t},k)$ if and only if the atom $p(\tuplevar{t})$ has a derivation of dimension $k$ in $P$.  

In this section we apply \emph{partial evaluation} \cite{Jones-Gomard-Sestoft}
to \emph{specialise} $P_{\mathit{dim}}$ with respect to dimension constraints. In particular, from a given set of CHCs $P$, and a dimension bound $k\ge 0$, we generate from $P_{\mathit{dim}}$ sets of clauses $P^{\atmost{k}}$ and $P^{\exceeds{k}}$, whose derivations
have dimension at most $k$ and at least $k+1$ respectively.

For instance, suppose
we wish to generate a set of clauses whose derivations for predicate $p$ have dimension at most 2.
Let $p(\tuplevar{x},k)$  be an atom and let $\phi(k)$ be a constraint restricting
the value of the dimension argument $k$, where in this case $\phi(k) \equiv k \le 2$. The goal of specialisation is to derive a set of clauses $P^{\atmost{2}}$, whose derivations for $p(\tuplevar{x},k)$ satisfy $\phi(k)$. 

Specialisation for this example could be achieved just by replacing each clause in $P_{\mathit{dim}}$ of the form $p(\tuplevar{x},k) \leftarrow \mathit{Body}$ by $p(\tuplevar{x},k) \leftarrow k \le 2 \wedge \mathit{Body}$ in $P^{\atmost{2}}$.
However, a derivation for $\mathit{Body}$ for which $k > 2$ gives an infeasible derivation for  $p(\tuplevar{x},k)$; we would like to eliminate as many such infeasible derivations as possible from $P^{\atmost{2}}$ by partially evaluating the atom $p(\tuplevar{x},k)$ and propagating the given constraint throughout the clauses.
The presence of clauses leading to infeasible derivations tends to cause analysis tools to make coarser approximations. Hence partial evaluation can increase the precision obtained when analysing or verifying dimension-constrained clauses.

\paragraph{Instantiation of a standard algorithm for partial evaluation.}There are many variants of partial evaluation algorithms for CHCs. 
We present here an instantiation of the
``basic algorithm" for partial evaluation of logic programs \cite{gallagher:pepm93}, which is parameterised by an ``unfolding rule" and an abstraction operation. 

The $\unfold_P$ operation is applied to a set of constrained facts $S$ representing goals, and returns a set of constrained facts representing subgoals obtained from the leaves of partial AND-trees for each element of $S$, constructed using the given unfolding rule.
More precisely, 
$$\begin{array}{lll}
\unfold_P(S) &=& \{ p_i(\tuplevar{x_i}) \leftarrow  (\phi \wedge \theta)\vert_{\tuplevar{x_i}} \mid \\
&& \hskip 1 cm  p(\tuplevar{x}) \leftarrow \theta \in S,\\
&& \hskip 1 cm  p(\tuplevar{x}) \leftarrow \phi,p_1(\tuplevar{x_1}),\ldots,p_m(\tuplevar{x_m}) \in P,\\
&& \hskip 1 cm \SAT(\theta \wedge \phi), \\
&& \hskip 1 cm 1 \le i \le m\}.
\end{array}
$$
$\phi\vert_{\tuplevar{v}}$ stands for the constraint $\exists \tuplevar{w}. \phi$, where $\tuplevar{w} = \vars(\phi) \setminus \tuplevar{v}$.

Given a set of constrained facts $S_0$ representing initial goals, the set 
 $\lfp \  \lambda S . (S_0 \cup \unfold_P(S))$ is the set of all constrained facts obtained from nodes in AND-trees for 
 elements of $S_0$.
 That is, if $p(\tuplevar{x}) \leftarrow \theta \in S_0$, $t$ is a feasible AND-tree with root labelled by $p(\tuplevar{x})$, and 
 $q(\tuplevar{y})$ 
 is the label of a node in $t$, then $q(\tuplevar{y}) \leftarrow (\constr(t) \wedge \theta)\vert_{\tuplevar{y}} \in \lfp \  \lambda S . (S_0 \cup \unfold_P(S))$.
 This set is usually infinite, and so we introduce an abstraction operation
$\abst_{\Psi}$  implementing a property-based abstraction \cite{DBLP:conf/tacas/GrebenshchikovGLPR12}.  This is based on a fixed set of constrained facts $\Psi$, and abstracts a set of constrained facts according to which properties in $\Psi$ they satisfy.  Formally,  $\abst_{\Psi}$ 
is defined as follows.
$$\begin{array}{lll}
\abst_{\Psi}(S) &=& \{ \rep_{\Psi}(p(\tuplevar{x}) \leftarrow \theta) \mid p(\tuplevar{x}) \leftarrow \theta \in S \}\\
&&\mathit{where}\\
\rep_{\Psi}(p(\tuplevar{x}) \leftarrow \theta) &=& p(\tuplevar{x}) \leftarrow \bigwedge \{ \psi \mid p(\tuplevar{x}) \leftarrow \psi \in \Psi  \wedge \theta \models \psi \}\\

\end{array}$$
Here, $\rep_{\psi}$ is applied to a constrained fact, returning its abstract ``representative" with respect to the set $\Psi$.
$\abst_{\Psi}(S)$ generalises the constrained facts in $S$;  for every constrained fact $p(\tuplevar{x}) \leftarrow \theta \in S$, there exists a (renamed) constrained fact $p(\tuplevar{x}) \leftarrow \phi \in \abst_{\Psi}(S)$ such that $\theta \models \phi$.
The maximum size of $\abst_{\Psi}(S)$ is $2^{\vert \Psi\vert}$ and so the closure $S^* = \lfp \  \lambda S . (S_0 \cup \abst_{\Psi}(\unfold_P(S)))$
is finite.

The partial evaluation algorithm returns a set of clauses,  $\renameunfold_{\Psi,P}(S^*)$.  The predicates in the clauses are renamed according to their versions;  that is, if $S^*$ contains constrained facts $p(\tuplevar{x}) \leftarrow \theta_1$ and $p(\tuplevar{x}) \leftarrow \theta_2$, then two renamed versions of $p$ are produced.  Formally, $\renameunfold_{\Psi,P}$ is defined as follows.
\[\begin{array}{ll}
\renameunfold_{\Psi,P}(S) = &\{ p^{v_0}(\tuplevar{x}) \leftarrow  \theta \wedge \phi,p_1^{v_1}(\tuplevar{x_1}),\ldots,p_m^{v_m}(\tuplevar{x_m}) \mid  \\
& \hskip 1cm p(\tuplevar{x}) \leftarrow \theta \in S,\\
& \hskip 1cm p(\tuplevar{x}) \leftarrow \phi,p_1(\tuplevar{x_1}),\ldots,p_m(\tuplevar{x_m}) \in P,\\
& \hskip 1cm \SAT(\theta \wedge \phi)\}\\
\end{array}\]
\noindent
where $p^{v_0}$ is the version of $p$ corresponding to its representative $\rep_{\Psi}(p(\tuplevar{x}) \leftarrow \theta)$ and for $1 \le j \le m$, $p_j^{v_j}$ is the version of $p$ corresponding to the representative $\rep_{\Psi}(p_j(\tuplevar{x_j}) \leftarrow (\theta \wedge \phi)\vert_{\tuplevar{x_j}})$.

As we will see in Sections \ref{pe-atmost} and \ref{pe-atleast}, partial evaluation can return clauses, whose dimension is bounded from above or below, depending on the initial set $S_0$ and the set $\Psi$.

\begin{proposition}\label{pe_correctness}
Let $P$ be a set of CHCs, $S_0$ and $\Psi$ be sets of constrained facts, where $S_0 \subseteq \Psi$.
Let $p(\tuplevar{x}) \leftarrow \theta(\tuplevar{x}) \in S_0$. 
Let $P'$ be the set of clauses $\renameunfold_{\Psi,P}(S^*)$ where $S^* = \lfp \  \lambda S . (S_0 \cup \abst_{\Psi}(\unfold_P(S)))$.

Then there exists a renamed version of $p$, say $p^{m}$ in $P'$ such that for all $\tuplevar{t}$, 
$P' \vdash p^{m}(\tuplevar{t})$ if and only if $P \vdash p(\tuplevar{t}) \wedge \theta(\tuplevar{t})$
\end{proposition} 
\begin{proof}
\begin{itemize}
\item
There exists $p^m$ such that $P \vdash p(\tuplevar{t}) \wedge \theta(\tuplevar{t})$ $\Rightarrow$ $P' \vdash p^m(\tuplevar{t})$. This follows from the soundness of 
the basic algorithm for partial evaluation, namely that it preserves the derivations that satisfy the input constraint. The proof is by induction on
the iterations of the computation of the fixpoint, and we do not give a proof here.
\item
There exists $p^m$ such that  $P' \vdash p^m(\tuplevar{t})$ $\Rightarrow$ $P \vdash p(\tuplevar{t}) \wedge \theta(\tuplevar{t})$.
By assumption, $p(\tuplevar{x}) \leftarrow \theta(\tuplevar{x}) \in \Psi$.  Hence $S^*$ contains a constrained fact  of the form
$p(\tuplevar{x}) \leftarrow \theta(\tuplevar{x}) \wedge \psi_1(\tuplevar{x}) \wedge \ldots \wedge \psi_j(\tuplevar{x})$, where $\{p(\tuplevar{x}) \leftarrow \theta(\tuplevar{x}), p(\tuplevar{x}) \leftarrow  \psi_1(\tuplevar{x}), \ldots p(\tuplevar{x}) \leftarrow  \psi_j(\tuplevar{x})\} \subseteq \Psi$ $(j \ge 0)$ and $\theta(\tuplevar{x}) \models \psi_i(\tuplevar{x})$, $0 \le i \le j$. 
Let $p^{m}$ be the renamed predicate corresponding to $p(\tuplevar{x}) \leftarrow \theta(\tuplevar{x}) \wedge \psi_1(\tuplevar{x}) \wedge \ldots \wedge \psi_j(\tuplevar{x})$.   For every clause with head $p^m(\tuplevar{x})$, its body contains $\theta(\tuplevar{x})$.

Let $u'$ be a feasible AND-tree for $p^m(\tuplevar{t})$ in $P'$. 
 By construction,
for every clause in $P'$, there is a clause in $P$ that is identical except for (a) predicate names and (b) the clause constraint, which is weaker in $P$ than in $P'$.  
Hence there is a feasible AND-tree $u$ for 
$p(\tuplevar{t})$ in $P$, that is identical to $u'$ except for predicate names, and $\constr(u') \rightarrow  \constr(u)$. 
Furthermore,  $\constr(u') \rightarrow \theta(\tuplevar{t})$ (as $\theta(\tuplevar{x})$ is in all clauses with head $p^m(\tuplevar{x})$), hence $\SAT(\theta(\tuplevar{t}) \wedge \constr(u))$. Hence there
is a feasible AND-tree for  $p(\tuplevar{t}) \wedge \theta(\tuplevar{t})$, that is, $P \vdash p(\tuplevar{t}) \wedge \theta(\tuplevar{t})$.
\end{itemize}

\end{proof}

\begin{example}
Let $P$ be the following clauses (containing no constraints in order to simplify the example).
 \begin{lstlisting}
p:- true.       p:- p,p.
\end{lstlisting}
\noindent
$P_{\mathit{dim}}$ is the following set of clauses, after unfolding the $\mathit{dim}$ predicates.
 \begin{lstlisting}
p(K):- K=0.
p(K):- p(K1), p(K2), K1>=K2+1, K=K1.
p(K):- p(K1), p(K2), K2>=K1+1, K=K2.
p(K):- p(K1), p(K2), K1=K2, K=K1+1.
\end{lstlisting}
\noindent
Let $\Psi$ in the   algorithm be $\{\texttt{p(K):-K=<1},\ \texttt{p(K):-K=<0}\}$ and $S_0 = \{\texttt{p(K):-K=<1}\}$. 
To compute $ \lfp \  \lambda S . (S_0 \cup \abst_{\Psi}(\unfold_P(S)))$, the algorithm computes sets $S_0, S_1,\ldots$ where $S_{i+1} = S_i ~\cup~ \abst_{\Psi}(\unfold_P(S_i))$. The $\lfp$ is the limit of this sequence, which is reached when
$S_{i+1} = S_i$. The key steps in the execution are as follows.
\begin{itemize}
\item $\unfold_P(S_0)$ first constructs the set of clauses in $P$ with $\texttt{K=<1}$ added to each body. For example, from the second clause in $P_{\mathit{dim}}$ we obtain:
 \begin{lstlisting}
p(K):- p(K1), p(K2), K1>=K2+1, K=K1, K=<1.
\end{lstlisting}
\noindent
The two constrained atoms in the above clause, after checking satisfiability and projecting the constraints onto their variables, are  
\[\{\texttt{p(K1):- K1=<1}, \ \ \texttt{p(K2):- K2=<0}\}\enspace .\]
\noindent
Applying $\abst_{\Psi}$ to this set yields $S_1 = $
\[\{\texttt{p(K):- K=<1}, \ \ \texttt{p(K):- K=<1,K=<0}\}\enspace .\]
\noindent
The other clauses are treated similarly but no other new constrained facts are returned.
\item Since $S_0 \neq S_1$, we compute $S_2 = S_1 \cup \abst_{\Psi}(\unfold_P(S_1))$. Since no new constrained facts are generated by 
this step (that is, $S_2 = S_1$),  the limit of the sequence is reached and so $S_2 =  \lfp \  \lambda S . (S_0 \cup \abst_{\Psi}(\unfold_P(S)))$.
\item $\renameunfold_{\Psi,P}(S_2)$ returns the following set of clauses.  The renaming distinguishes the two atoms in $S_2$,
renaming the predicate \texttt{p} as $\texttt{p\_1}$, corresponding to $\texttt{p(A,K):- K=<1,K=<0}$, and $\texttt{p\_2}$ 
corresponding to $\texttt{p(A,K):- K=<1}$.
 \begin{lstlisting}
p_2(B):- B=0.
p_2(B):- B>=F+1, B=<1, B=D, p_2(D), p_1(F).
p_2(B):- B>=D+1, B=<1, B=F, p_1(D), p_2(F).
p_2(B):- B=<1, B=D+1, B=F+1, p_1(D), p_1(F).
p_1(B):- B=0.
p_1(B):- B>=F+1, B=<0, B=D, p_1(D), p_1(F).
p_1(B):- B>=D+1, B=<0, B=F, p_1(D), p_1(F).
p_1(B):- B=<0, B=D+1, B=F+1, p_1(D), p_1(F).

\end{lstlisting}
\end{itemize}
\noindent
We notice that for predicate \texttt{p\_1} the last three clauses cannot succeed since they would yield a derivation whose dimension is greater than 0 and hence the constraints in those clauses would not be satisfied.  However, we can see that the successful derivations of \texttt{p\_1(K)} have \texttt{K=<0} and the successful derivations of \texttt{p\_2(K)} have \texttt{K=<1}. 
\end{example}

In Sections \ref{pe-atmost} and \ref{pe-atleast}, the partial evaluation algorithm is applied to $P_{\mathit{dim}}$ after first unfolding the $\mathit{dim}$ atoms (as shown in Figure \ref{fig:fibinstrumented-unfolded} for the clauses for \textit{Fib}), suitably instantiating the inputs $S_0$ and $\Psi$, to generate clauses whose derivations have dimensions that are bounded from above and below respectively.

\begin{figure}
	\centering
\begin{lstlisting}
fib(A,B,0) :- A>=0, A=<1, A=B. 
fib(A,B,K) :- A>1, D=A-2, E=A-1, B=F+G, 
          fib(D,G,K2), fib(E,F,K1), K1+1=<K, K2=K. 
fib(A,B,K) :- A>1, D=A-2, E=A-1, B=F+G, fib(D,G,K1), 
          fib(E,F,K2), K1+1=<K, K=K2. 
fib(A,B,K) :- A>1, D=A-2, E=A-1, B=F+G, 
          fib(D,G,K1), fib(E,F,K2),  K1=K-1, K2=K1.
\end{lstlisting}
\caption{Dimension instrumented \textit{Fib} program after unfolding dimension predicates.}
\label{fig:fibinstrumented-unfolded}
\end{figure}

\subsubsection{Construction of at-most-$k$-dimension set of clauses}\label{pe-atmost}

Given an instrumented set of CHCs $P_{\mathit{dim}}$ and $k \ge 0$, we apply the partial evaluation algorithm to obtain $P^{\atmost{k}}$, the at-most-$k$ dimension clauses. Let $\lhd \in \{ =, \leq\}$ and in the algorithm,  let $S_0 = \{ p(\tuplevar{x},z) \leftarrow z \lhd k \mid p ~\mathrm{is~a~predicate~in~}P\}$ and $\Psi = \{p(\tuplevar{x},z) \leftarrow z \lhd d \mid 0 \le d \le k,~p ~\mathrm{is~a~predicate~in~}P\}$.

Figure \ref{fig:pe-fib-atmost-1} shows the at-most-$1$-dimension clauses for \textit{Fib}.  The predicate names have been chosen to reflect the dimension constraints.  The final argument is the dimension, as in the instrumented clauses.
\begin{figure}[ht!]
\centering
\begin{lstlisting}
false$^{\exactly{1}}$(A):- C>5,C-D>0,A=1,fib$^{\exactly{1}}$(C,D,A).
false$^{\atmost{1}}$(A):- A>=0,C>5,C-D>0,-A>= -1,fib$^{\atmost{1}}$(C,D,A).

fib$^{\exactly{1}}$(A,B,C):- A>1,C=1,A-E=2,A-F=1,
                            B-G-H=0,I=0,fib$^{\exactly{1}}$(E,H,C),fib$^{\exactly{0}}$(F,G,I).
fib$^{\exactly{1}}$(A,B,C):- A>1,C=1,A-E=2,A-F=1,
                            B-G-H=0,I=0,fib$^{\exactly{0}}$(E,H,I),fib$^{\exactly{1}}$(F,G,C).
fib$^{\exactly{1}}$(A,B,C):- A>1,C=1,A-E=2,A-F=1,
                            B-G-H=0,I=0,fib$^{\exactly{0}}$(E,H,I),fib$^{\exactly{0}}$(F,G,I).
fib$^{\exactly{0}}$(A,B,C):- A>=0,-A>= -1,A-B=0,C=0.
fib$^{\atmost{1}}$(A,B,C):- A>=0,-A>= -1,A-B=0,C=0.
fib$^{\atmost{1}}$(A,B,C):- A>1,C=1,A-E=2,A-F=1,
                            B-G-H=0,I=0,fib$^{\exactly{1}}$(E,H,C),fib$^{\exactly{0}}$(F,G,I).
fib$^{\atmost{1}}$(A,B,C):- A>1,C=1,A-E=2,A-F=1,
                            B-G-H=0,I=0,fib$^{\exactly{0}}$(E,H,I),fib$^{\exactly{1}}$(F,G,C).
fib$^{\atmost{1}}$(A,B,C):- A>1,C=1,A-E=2,A-F=1,
                            B-G-H=0,I=0,fib$^{\exactly{0}}$(E,H,I),fib$^{\exactly{0}}$(F,G,I).

\end{lstlisting}
\caption{\textit{Fib}$^ {\atmost{1}}:$ at-most-\(1\)-dimension version of \textit{Fib}.}
 \label{fig:pe-fib-atmost-1}
\end{figure}

Note that for each predicate $p$ and each $d$, $0 \le d \le k$, the partial evaluation produces versions for both $p^{\atmost{d}}$ and
$p^{\exactly{d}}$ (though the set of clauses for some of these versions might be empty).

By Proposition \ref{pe_correctness}, for each predicate $p$ of $P$, $P^{\atmost{k}}$ contains a predicate (which we call $p^{\atmost{k}}$) all of whose derivations
have dimension at most $k$.

\subsubsection{Construction of at-least-$k$-dimension set of clauses}\label{pe-atleast}

We obtain $P^{\exceeds{k-1}}$, the at-least-$k$ dimension clauses in a similar way. In the algorithm, let $S_0 = \{ p(\tuplevar{x},z) \leftarrow z \ge k \mid p ~\mathrm{is~a~predicate~in~}P\}$. Let $\Psi = \{p(\tuplevar{x},z) \leftarrow z \ge d \mid 0 \le d \le k, ~p ~\mathrm{is~a~predicate~in~}P\}$.

Figure \ref{fig:pe-fib-atleast-1} shows the at-least-$1$-dimension clauses for \textit{Fib}.  The predicate names have been chosen to reflect the dimension constraints.

\begin{figure}[ht]
\centering
\begin{lstlisting}
false$^{\ge{1}}$(A):- A>=1,C>5,C-D>0,fib$^{\ge{1}}$(C,D,A).

fib$^{\ge{1}}$(A,B,C):- A>1,C-I>=1,I>=0,A-E=2,A-F=1,
                          B-G-H=0,fib$^{\ge{1}}$(E,H,C,J),fib$^{\anydim{0}}$(F,G,I,K).
fib$^{\ge{1}}$(A,B,C):- A>1,C-I>=1,I>=0,A-E=2,A-F=1,
                              B-G-H=0,fib$^{\anydim{0}}$(E,H,I),fib$^{\ge{1}}$(F,G,C).
fib$^{\ge{1}}$(A,B,C):- A>1,C>=1,A-E=2,A-F=1,
                         B-G-H=0,C-I=1,fib$^{\anydim{0}}$(E,H,I),fib$^{\anydim{0}}$(F,G,I).
fib$^{\anydim{0}}$(A,B,C):- A>=0,-A>= -1,A-B=0,C=0.
fib$^{\anydim{0}}$(A,B,C):- A>1,C-I>=1,I>=0,A-E=2,A-F=1,
                               B-G-H=0,fib$^{\ge{1}}$(E,H,C),fib$^{\anydim{0}}$(F,G,I).
fib$^{\anydim{0}}$(A,B,C):- A>1,C-I>=1,I>=0,A-E=2,A-F=1,
                               B-G-H=0,fib$^{\anydim{0}}$(E,H,I),fib$^{\ge{1}}$(F,G,C).
fib$^{\anydim{0}}$(A,B,C):- A>1,C>=1,A-E=2,A-F=1,
                          B-G-H=0,C-I=1,fib$^{\anydim{0}}$(E,H,I),fib$^{\anydim{0}}$(F,G,I).
\end{lstlisting}
\caption{\textit{Fib}$^ {\exceeds{0}}:$ at-least-\(1\)-dimension version of \textit{Fib}.}
 \label{fig:pe-fib-atleast-1}
\end{figure}

By Proposition \ref{pe_correctness}, for each predicate $p$ of $P$, $P^{\exceeds{k-1}}$ contains a predicate (which we call $p^{\ge{k}}$ or
sometimes $p^{\exceeds{k-1}}$) all of whose derivations
have dimension at least $k$.
 
%
%
%
%
%
%
%
%
%
%
%
%
%
%
%
%
%
%
%
%
%
%
%
%
%
%
%
%
%
%
%

 
%
%
\section{Verification Algorithms}
\label{sec:verification_algorithm}

In this section, we describe two algorithms for verification of CHCs, based on the notion of tree dimension. The verification problem we address is to decide whether a given set of CHCs has a model. In case it has no model, the problem is to find a counterexample.
A set of CHCs has a model if and only if there is no derivation of $\false$ from the clauses (or of $\false^{\atmost{k}} ~or~ \false^{\exceeds{k}}$ for some dimension bounded version of $\false$). Such a derivation exists only if the set contains at least one integrity constraint (clause with head $\false$ (or $\false^{\atmost{k}} ~or~ \false^{\exceeds{k}}$)).  
A set containing no integrity constraints has at least one model, namely the interpretation consisting of $p(x) \leftarrow \trueit$ for every predicate $p$ in the clauses.

\subsection{Decomposition by dimension of verification problem}

We first present an algorithm exploiting the decomposition of a set $P$ of CHCs into complementary sets $P^{\atmost{k}}$ and $P^{\exceeds{k}}$.
For each $k$, these two sets can be solved separately (possibly in parallel). 
\begin{proposition}[Decomposition by dimension] 
\label{decomposition-soundness}
A set of CHCs $P$ is safe if and only if for some $k$, both $P^{\atmost{k}}$ and $P^{\exceeds{k}}$ are safe.
\end{proposition}
\begin{proof}
Let  both $P^{\atmost{k}}$ and $P^{\exceeds{k}}$ be safe, for some $k$.  Equivalently, $P^{\atmost{k}} \not\vdash \false$ 
and $P^{\exceeds{k}} \not\vdash \false$. By the constructions in Sections \ref{pe-atmost} and
\ref{pe-atleast} and Proposition \ref{pe_correctness}, there is no derivation of $\false$ in $P$ of dimension $\le k$ or of dimension $> k$, which is equivalent to $P\not\vdash \false$, i.e. $P$ is safe.
\end{proof}
 
The essence of the algorithm based on tree dimension is to decompose $P$ into $P^{\atmost{k}}$ and $P^{\exceeds{k}}$ for successive values of $k$.   
 If for some $k$, both of them are safe, then $P$ is also safe, by Proposition \ref{decomposition-soundness}.  $P$ is unsafe if we find a $k$ such that one of them is unsafe.
 
\paragraph{Lifting interpretations.}
We introduce a lifting which constructs a syntactic interpretation for a set of CHCs given a syntactic interpretation for an annotated version of the same set of CHCs.

\begin{definition}[$S^{\uparrow}$: Lifting of an interpretation]
\label{liftingDefinition}
Let $Pred$ be a set of predicates, $I$ be a finite set.
Define $Pred^{I}=\{ p^\triangle \mid p\in Pred, \triangle \in I\}$ and let $S$ be an interpretation of $Pred^{I}$ given by constrained facts. 
Then $S^{\uparrow}$ is the following set of constrained facts:
\[
	S^{\uparrow} =   \{ p(x) \leftarrow {\textstyle \bigvee_{(p^{\triangle} (x)\leftarrow  \phi) \in S}} \ \phi 
	\mid p^\triangle \in Pred^{I}\}\enspace . 
\]
\end{definition}

The procedure SolvePartition defined in Algorithm \ref{alg:verify} makes use of a procedure $\Safe(P)$, which is a sound oracle: if it returns (\safe, \solution) then $P$ is safe; if it returns (\unsafe, \counterexample) then  $P$ is unsafe and the counterexample proves it; else we know nothing about $P$ and \stunknown\ is returned.
The oracle $\Safe$ could be any existing automatic Horn clause solver  \cite{DBLP:conf/tacas/GrebenshchikovGLPR12,DBLP:conf/cav/KafleGM16,DBLP:conf/sat/HoderB12,DBLP:conf/tacas/AngelisFPP14,LPAR-21:Synchronizing_Constrained_Horn_Clauses} possibly with a timeout.  
When it cannot verify a program within a given time limit, it returns \stunknown. 

Consider a call SolvePartition($P,k,\emptyset$), the algorithm checks first (using the oracle) whether $P^{\atmost{k}}$ is safe and if so then it proceeds to check whether $P^{\exceeds{k}}$ is safe. 
If, for either set of CHCs, the oracle returns \unsafe\ then the algorithm returns \unsafe. Similarly, if, for both sets of CHCs, the oracle returns \safe\  then the algorithm returns \safe\ together with the interpretation built so far augmented with the current \solution\ $R'$ (\emph{line 13}),  defining a model for $P$.
Otherwise, the $\Safe$ oracle returns \stunknown.  The \stunknown\ for $P^{\atmost{k}}$ is propagated and the \stunknown\ for
$P^{\exceeds{k}}$ causes the algorithm to proceed by calling itself recursively on the set $P^{\exceeds{k}}$, with $k+1$ and with the interpretation built so far.

\begin{algorithm}[h!]
\caption{SolvePartition($P$,$k$,$S$)\label{alg:verify}}
\begin{algorithmic}[1]
	\State \textbf{Input:} A set of CHCs $P$, an integer $k\geq 0$, and an interpretation $S$ (init $\emptyset$)
   \State \textbf{Output:} (\safe, \solution) $\mid$ (\unsafe, \counterexample) $\mid$ \stunknown
	 \State  $(\status, \mathit{R}) \gets \Safe$($P^{\atmost{k}}$) 
	\If {\status=\unsafe} 
	\State  \Return $(\unsafe\, \mathit{R}) $  \Comment{$P^{\atmost{k}}$ is \emph{unsafe}, hence $P$ is \emph{unsafe}} 
 	 \EndIf
	 \If {{\status=\stunknown} } 
	 \State  \Return \stunknown  \Comment{$P^{\atmost{k}}$ may be \emph{safe} or \emph{unsafe}, so is $P$} 
 	 \EndIf
	 \State $P^{\exceeds{ }} \gets  P^{\exceeds{k}}$ \Comment{We turn to $P^{\exceeds{k}}$ as $P^{\atmost{k}}$ is  \emph{safe}} 
    \State $(\status, \mathit{R'})  \gets \Safe(P^{\exceeds{ }})$
    \If {\status=\unsafe} 
    \State \Return $(\unsafe, \mathit{R'})$  \Comment{$P^{\exceeds{k}}$ is \emph{unsafe}, hence $P$ is  \emph{unsafe}} 
    \EndIf 
    \If {\status =\safe} 
    \State \Return $(\safe, (S\cup R\cup R')^{\uparrow})$  \Comment{$P^{\atmost{k}}$ and $P^{\exceeds{k}}$ are \emph{safe}, hence $P$ is \emph{safe}} 
    \EndIf 
    \State \Return SolvePartition($P^{\exceeds{}}$,$k+1$,$S\cup R$) \Comment{recurse: $P^{\exceeds{k}}$ may be \emph{safe} or \emph{unsafe}}
\end{algorithmic}
\end{algorithm}

\begin{example}
Applying the algorithm to our example program \emph{Fib}, the oracle $\Safe$ finds  that  both $Fib^{\atmost{0}}$  and $Fib^{\exceeds{0}}$  are safe, and thus  \emph{Fib} is safe. 
\end{example}

The soundness of the above algorithm follows from the soundness of the oracle and the properties of dimension bounded set of clauses, which is formally stated by the following proposition.

\begin{proposition}[Soundness] 
  If Algorithm \ref{alg:verify} returns  (\emph{un})\emph{safe} on a set of CHCs $P$ then $P$ is (un)safe.
\end{proposition}

 %
%
\subsection{Verification by successive iteration of bounded dimension CHCs}

For an unsafe program $P$,  there exists some $k_0 \geq 0$ such that $P^{\atmost{k}}$ is \emph{unsafe} for all $k \geq k_0$. 
So for discovering a bug, we can  generate $P^{\atmost{k}}$ successively for $k=0,1,2, \dots$ and check its safety,  as in bounded model checking (BMC). 
In BMC, the original program and the bounded underapproximations are decidable. By contrast, the under-approximations obtained by dimension bounding are themselves undecidable and there is no upper bound on the dimension.

However, a solution of a bounded dimension program can extend to a solution of the original problem as we shall see.  The example in Figure \ref{ex:revlen} is the at-most-1-dimension version of the example in Figure \ref{ex:revlen}. The oracle \Safe\ derives the following invariant for the predicates $\mathtt{applen^{\atmost{1}}}$ and $\mathtt{revlen^{\atmost{1}}}$ from it. This invariant (mapped to the original program using Definition \ref{liftingDefinition}) is in fact an invariant of the original program (in Figure \ref{ex:revlen}).  Thus the solution of an underapproximation is the solution of the original program. 
\[
\begin{aligned}
& \mathtt{applen^{\atmost{1}}(A,B,C) \leftarrow B \geq 0 \wedge A\geq 0 \wedge A+B=C}.  \\
& \mathtt{revlen^{\atmost{1}}(A,B) \leftarrow B\geq0 \wedge A=B}.
\end{aligned}
\]

\begin{figure}[t]
\begin{lstlisting}
applen(A,B,C):- A=0, B=C, B>=0.
applen(A,B,C):- applen(A1,B,C1), A=A1+1, C=C1+1.

revlen(A,B):- A=0, B=0.
revlen(A,B):- revlen(A1,C), applen(C,D,B), A=A1+1, D=1.

false :- revlen(A,B), A$\neq$B. 
\end{lstlisting} 
\caption{Length abstracted version of reverse of a list.}
 \label{ex:revlen}
\end{figure}

\begin{figure}[t]
\begin{lstlisting}
applen$^{=0}$(A,B,C):- A=0,   B=C,  B>=0.
applen$^{=1}$(A,B,C):- A=D+1, C=E+1, applen$^{=1}$(D,B,E).
applen$^{=0}$(A,B,C):- A=D+1, C=E+1, applen$^{=0}$(D,B,E).
applen$^{\atmost{1}}$(A,B,C):- applen$^{=1}$(A,B,C). 
applen$^{\atmost{1}}$(A,B,C):- applen$^{=0}$(A,B,C).
applen$^{\atmost{0}}$(A,B,C):- applen$^{=0}$(A,B,C).

revlen$^{=0}$(A,B) :- A=0,  B=0.
revlen$^{=1}$(A,B) :- A=C+1, E=1,  applen$^{\atmost{0}}$(D,E,B), revlen$^{=1}$(C,D).
revlen$^{=1}$(A,B) :- A=C+1,  E=1, revlen$^{\atmost{0}}$(C,D),  applen$^{=1}$(D,E,B).
revlen$^{=1}$(A,B) :-  A=C+1,  E=1, revlen$^{=0}$(C,D),  applen$^{=0}$(D,E,B).
revlen$^{\atmost{1}}$(A,B) :- revlen$^{=1}$(A,B).  revlen$^{\atmost{1}}$(A,B) :- revlen$^{=0}$(A,B).
revlen$^{\atmost{0}}$(A,B) :- revlen$^{=0}$(A,B).   

false$^{=1}$:- A$\neq$B,  revlen$^{=1}$(A,B).  false$^{=0}$ :- A$\neq$B, revlen$^{=0}$(A,B).
false$^{\atmost{1}}$ :- false$^{=1}$. false$^{\atmost{1}}$ :- false$^{=0}$. false$^{\atmost{0}}$ :- false$^{=0}$.

\end{lstlisting} 
\caption{At-most-1-dim version of reverse list example in Figure \ref{ex:revlen}}
 \label{revlen-1dim}
\end{figure}

Therefore,  the safety of $P^{\atmost{k}}$ also can be checked successively for increasing value of $k$ starting from \(0\) until $P^{\atmost{k}}$ (for some $k$) is proven \unsafe~ or its solution generalises to $P$ or the results for $P^{\atmost{k}}$ is \stunknown. A solution for $P^{\atmost{k}}$ generalises to $P$ if the lifted model of $P^{\atmost{k}}$ using Definition \ref{liftingDefinition} is also a model of $P$. The iteration done  this way does not make any reuse of solutions of lower dimension while verifying a program of higher dimension, which could save some verification effort. The iteration in which the iterates of higher dimension  reuses solutions from lower dimensions is reminiscent of  \emph{Newtonian iteration} \cite{DBLP:journals/jacm/EsparzaKL10}. However, reuse introduces a new problem since solutions are approximate.  If a counterexample is found for $P^{\atmost{k+1}}$ (where solutions from lower dimensions are used), it  needs to be further examined since it may not be a counterexample for $P^{\atmost{k+1}}$ (in which no solutions are used from lower dimensions).  We present a solution to this problem in  Algorithm \ref{alg:verify3} via refinement of approximations. 
Before presenting the algorithm, we first introduce Definition \ref{cexprojection} which defines the subset of $S$ of constrained facts involved in a derivation $t$.  

\begin{definition}[$S_{|t}$]
\label{cexprojection}
Let $S$ be an interpretation of a set of CHCs $P$ given by constrained facts and let $t$ be any derivation in $P$. Define $S_{|t}$ to be 
\[S_{|t}=\{ (A \leftarrow \phi) \mid  (A \leftarrow \phi) \in S \wedge \text{atom $A$ labels a node of $t$} \} \enspace.\]
\end{definition}

Next, we define the auxiliary procedure \textsf{subst}() as follows.
\begin{definition}[\textsf{subst}($P$,$S$)]
\label{def:subst}
Given a set of CHCs $P$ and an interpretation $S$, define
\textsf{subst}(\(P\),\(S\)) as the set of CHCs obtained as follows:  for every constrained fact $A \leftarrow \phi$ in $S$, replace all the clauses from $P$ whose head is $A$ with the clause $A \leftarrow \phi$.
\end{definition}

We now turn to Algorithm~\ref{alg:verify3}. Consider a call SolveInc($P,k,\emptyset$); the algorithm checks first (using the oracle) whether $P^{\atmost{k}}$ with the information provided by $S$ “plugged in” using \textsf{subst} is safe (\emph{line 3}). 
If the oracle returns \stunknown, the algorithm returns \stunknown~(\emph{line 4-5}). 
Else if the oracle returns \unsafe, the \counterexample~$R$ is further examined (\emph{line 6-9}). 
If it uses no constrained facts of $S$ then the counterexample is also a counterexample for $P$ (\emph{line 7-8}).
In the case that some constrained facts of $S$ are used in $R$ then the algorithm recurses with those facts removed from $S$ (\emph{line 9}). 
Finally, if the oracle returns \safe~the algorithm checks whether the model extends to $P$ and returns \safe~if so (\emph{line 10-11}).
Should the check fail  the algorithm recurses with $k$
increased (\emph{line 12}).

Removing the over-approximations used by the counterexample ensures progress (as we shall see in the example below) in the sense that the same counterexample does not arise again in the next iteration. 
This is because if the same trace arises again and does not use any over-approximations, then it must be a counterexample.  
In the worst case, all the solutions from the lower dimensions are removed.

Consider  an example program (linear for simplicity) shown below. 
 
\begin{lstlisting}
 c1. false:- X=0, p(X).              c2. false:- q(X).
 c3. p(X):- X>0.                     c4. q(X):- X=0.
\end{lstlisting}

 Suppose we have an approximate solution \texttt{$S=\{ p(X) \leftarrow \mathit{true}\}$} for the predicate $p(X)$. Using this solution, the above program is transformed into  the following  program.
 
  \begin{lstlisting}
 c1. false:- X=0, p(X).           c2. false:- q(X).
 c3. p(X):- true. 
 c4. q(X):- X=0.
 \end{lstlisting}

The trace \texttt{c1(c3)} is a counterexample for this transformed program  but not  for the original program (since it uses an approximate solution for the predicate $p$). However  the trace \texttt{c2(c4)} is a counterexample for this program as well as for the original  since it  does not use any approximate solution for the predicates appearing in the counterexample.
 
 \algnewcommand{\algorithmicgoto}{\textbf{go to}}%
\algnewcommand{\Goto}[1]{\algorithmicgoto~\ref{#1}}%
 \begin{algorithm}[t!]
\caption{SolveInc($P$, $k$, $S$)\label{alg:verify3}}
\begin{algorithmic}[1]
   \State \textbf{Input:} A set of CHCs $P$, an integer $k\geq 0$, and an interpretation $S$ (init $\emptyset$)
   \State \textbf{Output:} (\safe, \solution) $\mid$ (\unsafe, \counterexample) $\mid$ \stunknown
   \State $(\status, \mathit{R})  \gets \Safe(\Subst(P^{\atmost{k }},S))$\Comment{substitute syntactic model $S$ into $P^{\atmost{ }}$}\label{marker}
   \If {\status =\stunknown }
   \State \Return \stunknown
   \EndIf
   \If{\status =\unsafe }\Comment{$R$ is a counterexample for $\Subst(P^{\atmost{ }},S)$}
        \If{ \(\mathit{S}_{ | R} = \emptyset \)} \Comment{$R$ uses no predicate defined by $S$}
	\State \Return \( (\unsafe, \mathit{R})\) \Comment{ hence $R$ is a  counterexample for $P$}
        \EndIf
	\State \Return SolveInc($P$, $k$, $S \setminus \mathit{S}_{ | R}$) \Comment{recurse with the facts of $S$ not used in $R$}
    \EndIf
   \If {($\mathit{R}^{\uparrow}$  \text{is a  solution  of } $P$)} \Comment{$\Subst(P^{\atmost{k}},S)$ is safe}%
          \State \Return \((\safe,\mathit{R}^{ \uparrow})\) 
   \EndIf
   \State \Return SolveInc($P$, $k+1$, $R$) \Comment{$R^\uparrow$ does not solve $P$, recurse}
\end{algorithmic}
\end{algorithm}


%
%
\section{Experimental results }
\label{experiments}
 
\subsection{Verification of safety properties}

\paragraph{\textbf{Implementation and experimental setting.}} Algorithms \ref{alg:verify} and    \ref{alg:verify3}   are implemented in Ciao Prolog \cite{DBLP:journals/tplp/HermenegildoBCLMMP12}, interfaced with the Parma Polyhedra Library \cite{DBLP:journals/scp/BagnaraHZ08} and the Yices 2.2  SMT solver \cite{Dutertre:cav2014} for  the  manipulation of constraints.  The  experiments are carried out on a set of 45 (36 safe and 9 unsafe) CHC verification problems taken from three sources:  the repository of NTS  benchmarks\footnote{\url{https://github.com/pierreganty/NTSLib/}}, the recursive category of SV-COMP\footnote{\url{http://sv-comp.sosy-lab.org/2015/benchmarks.php}} \cite{DBLP:conf/tacas/000115} and the benchmarks from the QARMC tool \cite{DBLP:conf/tacas/GrebenshchikovGLPR12}. 
Examples were chosen that  potentially have derivations of unbounded dimension (that is, they are  sets of non-linear clauses). 
Some of these benchmarks are first translated to Prolog syntax  using the tools ELDARICA\footnote{\url{https://github.com/uuverifiers/eldarica}} \cite{DBLP:conf/fm/HojjatKGIKR12} and SeaHorn \cite{DBLP:conf/tacas/GurfinkelKN15}. The benchmarks  are not beyond the capabilities of the existing Horn clause solvers, but they are typically used for testing the performance of new tools. The experimental evaluation is done on a MacBook Pro  running  OS X on  2.3 GHz Intel core i7 processor, 4 cores and 8 GB memory.   The results of these algorithms are compared with that of \rahft\ \cite{DBLP:conf/cav/KafleGM16},  a Horn clause verifier which refines an abstract interpretation by eliminating infeasible derivations.

 \paragraph{Implementation of $P^{\atmost{k}}$ and $P^{\exceeds{k}}$.} 
 For the experiments, we constructed the set of clauses $P^{\atmost{k}}$ and $P^{\exceeds{k}}$ using the procedures described in Section \ref{sec:pe}.

The experiments are intended to establish (i) whether the dimension-based decomposition is practical,  (ii) the relationship between the dimension and the solvability of a problem and (iii) how this approach compares other approaches.

\paragraph{\textbf{Discussion.} }The results are summarised in Table~\ref{tbl:exp}. We report results for three verification algorithms, namely Algorithm \ref{alg:verify} and Algorithm \ref{alg:verify3}, and the \Safe\ oracle that is used in those algorithms.
For Algorithm \ref{alg:verify}, we report the result returned, the dimension bound, and the time.  For the \Safe\ oracle we report the result and time, and for Algorithm \ref{alg:verify3} we return the dimension when a result is returned (if at all) and the time.

The oracle \Safe\ used in both algorithms is an abstract interpreter over the domain of convex polyhedra \cite{DBLP:conf/cav/KafleGM16}, which returns \unsafe\ if a feasible derivation of the predicate $\false$ exists,  \safe\ if a syntactic model can be found within a time bound, and \stunknown\ otherwise.

\begin{table}
 \caption{Experimental  results on 45 CHC verification problems with a timeout of 5 minutes. 
 Times are in seconds. \newline
 \label{tbl:exp}}
 \begin{minipage}{\textwidth}
    \begin{tabular}{|l|l|l|l|l|l|l|l|}  
     \cline{1-8}
                              &       &  \multicolumn{2}{c|}{\bf{Alg. \ref{alg:verify}}}          & \multicolumn{2}{c|}{\bf{\Safe\ oracle}}             &  \multicolumn{2}{c|}{\bf{Alg. \ref{alg:verify3}}}  
   \\ \cline{1-8}
    Program                    & safety & dim         &time&result & time & dim &time

     \\ \cline{1-8}
    Addition03\_false-unreach   & safe   & 0           & 3        & safe       & $< 1$  & ? & ? \\ \cline{1-8}
    McCarthy91\_false-unreach   & unsafe & 1           & 6        & unsafe        & $< 1$  & ?           & ?        \\ \cline{1-8}
    addition.nts.pl             		& safe   & 0           & 3        & safe        & $< 1$  & 1           & $<1$   \\ \cline{1-8}
    bfprt.nts.pl               		 & safe   & 0           & 5        & safe        & $< 1$  & 2           & 4        \\ \cline{1-8}
    binarysearch.nts.pl         	& safe   & 0           & 3        & safe        & $< 1$  & 1           & 1.1      \\ \cline{1-8}
    countZero.nts.pl           		 & safe   & 0           & 4        & safe        & $< 1$  & 1           & $< 1$  \\ \cline{1-8}
    eq.horn                     		& unsafe & 0           & 3        & unsafe        & $< 1$  & 2           & $< 1$  \\ \cline{1-8}
    fib.pl                      			& safe   & 1           & 6        & ?        & ? & ?           & ?        \\ \cline{1-8}
    identity.nts.pl             		& safe   & 0           & 4        & safe        & $< 1$  & 1           & $< 1$  \\ \cline{1-8}
    merge.nts.pl               		 & safe   & 0           & 5        & safe        & $< 1$  & 1           & 1.7      \\ \cline{1-8}
    palindrome.nts.pl           	& safe   & 0           & 3        & safe        & $< 1$  & 1           & $< 1$  \\ \cline{1-8}
    parity.nts.pl               		& unsafe & 0           & 3        & ?        & ?	  & ?           & ?        \\ \cline{1-8}
    remainder.nts.pl            		& unsafe & 0           & 3        & unsafe       & $< 1$  & 1           & $< 1$  \\ \cline{1-8}
    revlen.pl                   		& safe   & 0           & 3        & safe        & $< 1$  & 1           & $< 1$  \\ \cline{1-8}
    running.nts.pl              		& unsafe & 1           & 4       & ?        & ?  & ?           & ?        \\ \cline{1-8}
    sum\_10x0\_false-unreach    & unsafe & ?           & ?        & ?       & ?       & ?           & ?        \\ \cline{1-8}
    sum\_non\_eq\_false-unreach & unsafe & 0           & 3        & unsafe        & $< 1$  & ?           & ?        \\ \cline{1-8}
    suma1.horn                  		& unsafe & 0           & 3        & unsafe        & $< 1$  & 1           & $< 1$  \\ \cline{1-8}
    suma2.horn                  		& unsafe & 0           & 3       & unsafe        & $< 1$  & 2           & $< 1$  \\ \cline{1-8}
    summ\_SG1.r.horn            	& safe   & 0           & 2        & safe        & $< 1$  & ?           & ?        \\ \cline{1-8}
    summ\_SG2.r.horn            	& safe   & ?           & ?        & ?        & ?       & ?           & ?        \\ \cline{1-8}
    summ\_SG3.horn              	& safe   & 0           & 3        & safe        & $< 1$  & 1           & $< 1$  \\ \cline{1-8}
    summ\_b.horn                	& safe   & 2           & 12        & ?       & ?     & ?           & ?        \\ \cline{1-8}
    summ\_binsearch.horn        	& safe   & ?           & ?       & ?        & ?       & ?           & ?        \\ \cline{1-8}
    summ\_cil.casts.horn        	& safe   & 0           & 3        & safe        & $< 1$  & 1           & $< 1$  \\ \cline{1-8}
    summ\_formals.horn         	 & safe   & 0           & 4        & safe        & $< 1$  & 1           & $< 1$  \\ \cline{1-8}
    summ\_g.horn                	& safe   & 0           & 3        & safe        & $< 1$  & ?           & ?        \\ \cline{1-8}
    summ\_globals.horn          	& safe   & 0           & 3        & safe        & $< 1$  & 1           & $< 1$  \\ \cline{1-8}
    summ\_h.horn                	& safe   & 0           & 3        & safe        & $< 1$  & 2           & $< 1$  \\ \cline{1-8}
    summ\_local-ctx-call.horn   	& safe   & 0           & 2       & safe        & $< 1$  & 1           & $< 1$  \\ \cline{1-8}
    summ\_locals.horn           	& safe   & 0           & 4        & safe        & $< 1$  & ?           & ?        \\ \cline{1-8}
    summ\_locals2.horn          	& safe   & 0           & 2        & safe        & $< 1$  & 1           & $< 1$  \\ \cline{1-8}
    summ\_locals3.horn          	& safe   & 0           & 3        & safe        & $< 1$  & 1           & $< 1$  \\ \cline{1-8}
    summ\_locals4.horn          	& safe   & 0           & 3        & safe        & $< 1$  & 2           & 2.2      \\ \cline{1-8}
    summ\_mccarthy2.horn        & safe   & ?           & ?        & ?        & ?        & ?           & ?        \\ \cline{1-8}
    summ\_multi-call.horn       	& safe   & 0           & 3        & safe        & $< 1$  & 1           & $< 1$  \\ \cline{1-8}
    summ\_nested.horn           	& safe   & 0           & 3        & safe        & $< 1$  & 1           & $< 1$  \\ \cline{1-8}
    summ\_ptr\_assign.horn      	& safe   & 0           & 3        & safe        & $< 1$  & 1           & $< 1$  \\ \cline{1-8}
    summ\_recursive.horn        	& safe   & 0           & 3        & ?        & ?        & ?           & ?        \\ \cline{1-8}
    summ\_rholocal.horn         	& safe   & 0           & 3        & safe        & $< 1$  & 1           & $< 1$  \\ \cline{1-8}
    summ\_rholocal2.horn       	 & safe   & 0           & 3        & safe        & $< 1$  & 1           & $< 1$  \\ \cline{1-8}
    summ\_slicing.horn          	& safe   & 0           & 3        & safe        & $< 1$  & ?           & ?        \\ \cline{1-8}
    summ\_summs.horn            	& safe   & 0           & 3        & safe        & $< 1$  & ?           & ?        \\ \cline{1-8}
    summ\_typedef.horn          	& safe   & 0           & 4        & safe        & $< 1$  & 1           & $< 1$  \\ \cline{1-8}
    summ\_x.horn                	& safe   & 0           & 3        & safe        & $< 1$  & ?           & ?        \\  \cline{1-8}
 \cline{1-8}
        \#  solved (safe/unsafe)                     & ~      & \multicolumn{2}{c|}{43 (35/8)}         & \multicolumn{2}{c|}{36 (30/6)}    & \multicolumn{2}{c|}{27 (23/4)}  \\ \cline{1-8}
    \end{tabular}
    \vspace{-2\baselineskip}
 \end{minipage}
\end{table}

Firstly, the results show that implementation of decomposition based on tree dimension is practical.
Algorithm \ref{alg:verify} solves 43 out of 45 problems and Algorithm \ref{alg:verify3}  solves about 27 out of 45 problems.
There are 7 examples where Algorithm \ref{alg:verify} with dimension $k\le2$ was enough to prove safety but the \Safe\ oracle was not able to return \safe\ or \unsafe.  
That is, with a given oracle \Safe, there are examples for which \Safe\ returns unknown on the original clauses, but there is a low dimension (say $k=0$ or $k=1$) where \Safe\ returns safe on both $P^{\atmost{k}}$ and $P^{\exceeds{k}}$. This is evidence that decomposition by dimension is useful with respect to that particular oracle and is an effective refinement heuristic for these cases.   While in other refinement approaches, a spurious counterexample is the basis of refinement, Algorithms \ref{alg:verify} and \ref{alg:verify3} can be  viewed as performing refinement in which clauses are refined by eliminating safe derivations of lower dimensions, thereby removing a possibly infinite number of traces that have already been shown to be safe.  

Most of the problems solved using Algorithm \ref{alg:verify} are solved when they are decomposed with dimension  $k=0$. The separation of the derivations ($k=0$) eases the verification task. Only 4 problems that were solved needed decomposition greater than \(0\). Similarly, for Algorithm \ref{alg:verify3}, the solution of an under-approximation ($P^{\atmost{k}}$) for a fairly small value of $k=1 ~\mathrm{or}~ k=2$ was sufficient for finding a syntactic model for those problems that were solved. Though this observation may be related to this particular set of examples, we suspect that many application problems resulting from encoding imperative programs have derivation trees  of low dimension. 

\paragraph{Use of linearisation.}
As mentioned previously, $P^{\atmost{k}}$ can be linearised, potentially allowing the use of specialised verification procedures for linear clauses.  Although our implementation of the oracle \Safe\ contains no special facilities for dealing with linear clauses, we applied a linearisation procedure in Algorithm \ref{alg:verify3}. For this purpose we used a procedure based on partial evaluation \cite{DBLP:journals/corr/KafleGG16}.  We did not observe that linearisation in itself offers any advantages, although one might expect that linear clauses were in some way a simpler case for verification.  In order to exploit linearisation, it would be necessary to use a verification procedure with more specialised procedures for recognising and solving particular classes of linear recursive predicates amenable to precise solution, for example as described by  \citeN{DBLP:conf/sas/GonnordH06}.

\paragraph{Limitations of our experiments and possible improvements.} 
It would be possible to combine dimension-bounded decomposition with refinement-based solving.  For example, our oracle \Safe\ is limited in that it does not attempt any refinement after computing a convex polyhedra abstract interpretation of the clauses. It returns  \stunknown\ if the over-approximation  allows a derivation of $\false$, which might, however, be infeasible. Thus Algorithm \ref{alg:verify3} tends to return \stunknown\ before the timeout, in cases where a more sophisticated oracle would allow the procedure to continue.


%
%
\section{Discussion and Related Work}
\label{rel}

The notion of dimension of a tree has a long history in science (starting with
Geology) which has been detailed by 
\citeN{DBLP:conf/lata/EsparzaLS14}. However, the use of dimension for program verification is
more recent.  \citeN{Ganty2016} used the notion of
\emph{tree dimension} for computing summaries of procedural programs by
under-approximating them. Roughly speaking, they compute procedure summaries
iteratively, starting from the program behaviours captured by derivation trees
of dimension \(0\). Then they reuse these summaries to compute summaries for
program behaviours captured by derivation trees of dimension \(1\) and so on for
\(2\), \(3\), etc. We adapt the idea of dimension-based
under-approximations to the setting of CHCs.

Decomposition can be compared to refinement techniques based on automata \cite{DBLP:conf/sas/HeizmannHP09,DBLP:conf/cav/HeizmannHP13,kafleG2015Horn} in which the aim is to eliminate sets of program traces that have been shown to be safe.  In our case, establishing the safety of clauses whose derivations are of a given dimension allows us to eliminate those dimensions, and focus on the remaining dimensions.  Our decomposition technique offers a practical way to checking and eliminating infinite sets of traces.

In the world of constrained Horn clause verification tools (solvers)  we can
distinguish solvers depending on whether they can handle general non-linear Horn
clauses or not. A majority of solvers
\cite{DBLP:conf/tacas/GurfinkelKN15,DBLP:conf/tacas/GrebenshchikovGLPR12,DBLP:conf/cav/RummerHK13,McmillanR2013,kafleG2015Horn}
handle non-linear Horn clauses but there are notable exceptions like
VeriMAP~\cite{DBLP:conf/tacas/AngelisFPP14} or
Sally\footnote{\url{https://github.com/SRI-CSL/sally}}. For both VeriMAP and Sally, their
underlying reasoning engine handles only linear Horn clauses which appears to restrict,
in principle, their applicability. 
However, prior work on consistency preserving linearisation of dimension-bounded sets of CHCs \cite{DBLP:journals/corr/KafleGG16} shows solvers for linear CHCs can be used to check consistency of non-linear sets of CHCs.
Another work on linearisation of CHCs
based on \emph{fold-unfold transformations} is described by \citeN{DBLP:journals/tplp/AngelisFPP15}.
 
%
%
\section{Conclusion and Future Work}
\label{concl}

We applied   the notion of \emph{tree
dimension} to decompose  constrained Horn clause verification  problems by  dimensions.  We presented algorithms based on this idea; whose results  on a set of non-linear Horn clause verification benchmarks show its feasibility and usefulness both for proving safety  as well as for finding bugs in programs. We also looked into the problem of instrumenting clauses with dimension predicates and reason about the dimension directly from the resulting clauses.

Other ideas for program verification based on tree dimension are worth investigating, including induction based on tree dimension,  and further investigation of strategies that could exploit knowledge of dimension bounds (such as those discussed in Section \ref{sec:proginstr}).

 
\section*{Acknowledgements}
The research leading to these results has been supported by EU FP7 project 318337, ENTRA - Whole-Systems Energy
Transparency, EU FP7 project 611004, coordination and support action ICT-Energy, EU FP7 project 610686, POLCA - Programming
Large Scale Heterogeneous Infrastructures, Madrid Regional Government project S2013/ICE-2731, N-Greens Software
- Next-GeneRation Energy-EfficieNt Secure Software, and the Spanish Ministry of Economy and Competitiveness project
No. TIN2015-71819-P, RISCO - RIgorous analysis of Sophisticated COncurrent and distributed systems. The first author is supported by the Australian Research Council Discovery Project grant DP140102194.

We thank the anonymous reviewers for their comments and suggestions, which greatly improved the paper.

\bibliographystyle{acmtrans}

\begin{thebibliography}{}

\bibitem[\protect\citeauthoryear{Abelson and Sussman}{Abelson and
  Sussman}{1996}]{DBLP:books/mit/AbelsonS96}
{\sc Abelson, H.} {\sc and} {\sc Sussman, G.~J.} 1996.
\newblock {\em Structure and Interpretation of Computer Programs, Second
  Edition}.
\newblock {MIT} Press.

\bibitem[\protect\citeauthoryear{Afrati, Gergatsoulis, and Toni}{Afrati
  et~al\mbox{.}}{2003}]{DBLP:journals/tcs/AfratiGT03}
{\sc Afrati, F.~N.}, {\sc Gergatsoulis, M.}, {\sc and} {\sc Toni, F.} 2003.
\newblock Linearisability on datalog programs.
\newblock {\em Theor. Comput. Sci.\/}~{\em 308,\/}~1-3, 199--226.

\bibitem[\protect\citeauthoryear{Bagnara, Hill, and Zaffanella}{Bagnara
  et~al\mbox{.}}{2008}]{DBLP:journals/scp/BagnaraHZ08}
{\sc Bagnara, R.}, {\sc Hill, P.~M.}, {\sc and} {\sc Zaffanella, E.} 2008.
\newblock The {P}arma {P}olyhedra {L}ibrary: Toward a complete set of numerical
  abstractions for the analysis and verification of hardware and software
  systems.
\newblock {\em Sci. Comput. Program.\/}~{\em 72,\/}~1-2, 3--21.

\bibitem[\protect\citeauthoryear{Baier and Tinelli}{Baier and
  Tinelli}{2015}]{DBLP:conf/tacas/2015}
{\sc Baier, C.} {\sc and} {\sc Tinelli, C.}, Eds. 2015.
\newblock {\em {TACAS}. Proceedings}. LNCS, vol. 9035. Springer.

\bibitem[\protect\citeauthoryear{Beyer}{Beyer}{2015}]{DBLP:conf/tacas/000115}
{\sc Beyer, D.} 2015.
\newblock Software verification and verifiable witnesses - (report on {SV-COMP}
  2015).
\newblock See \citeN{DBLP:conf/tacas/2015}, 401--416.

\bibitem[\protect\citeauthoryear{Bj{\o}rner, Gurfinkel, McMillan, and
  Rybalchenko}{Bj{\o}rner
  et~al\mbox{.}}{2015}]{DBLP:conf/birthday/BjornerGMR15}
{\sc Bj{\o}rner, N.}, {\sc Gurfinkel, A.}, {\sc McMillan, K.~L.}, {\sc and}
  {\sc Rybalchenko, A.} 2015.
\newblock {H}orn clause solvers for program verification.
\newblock In {\em Fields of Logic and Computation {II} - Essays Dedicated to
  Yuri Gurevich on the Occasion of His 75th Birthday}, {L.~D. Beklemishev},
  {A.~Blass}, {N.~Dershowitz}, {B.~Finkbeiner}, {and} {W.~Schulte}, Eds.
  {LNCS}, vol. 9300. Springer, 24--51.

\bibitem[\protect\citeauthoryear{Bj{\o}rner, McMillan, and
  Rybalchenko}{Bj{\o}rner et~al\mbox{.}}{2013}]{DBLP:conf/sas/BjornerMR13}
{\sc Bj{\o}rner, N.}, {\sc McMillan, K.~L.}, {\sc and} {\sc Rybalchenko, A.}
  2013.
\newblock On solving universally quantified {H}orn clauses.
\newblock In {\em SAS}, {F.~Logozzo} {and} {M.~F{\"a}hndrich}, Eds. LNCS, vol.
  7935. Springer, 105--125.

\bibitem[\protect\citeauthoryear{De~Angelis, Fioravanti, Pettorossi, and
  Proietti}{De~Angelis et~al\mbox{.}}{2014}]{DBLP:conf/tacas/AngelisFPP14}
{\sc De~Angelis, E.}, {\sc Fioravanti, F.}, {\sc Pettorossi, A.}, {\sc and}
  {\sc Proietti, M.} 2014.
\newblock Verimap: A tool for verifying programs through transformations.
\newblock In {\em TACAS}, {E.~{\'A}brah{\'a}m} {and} {K.~Havelund}, Eds. LNCS,
  vol. 8413. Springer, 568--574.

\bibitem[\protect\citeauthoryear{{De Angelis}, Fioravanti, Pettorossi, and
  Proietti}{{De Angelis} et~al\mbox{.}}{2015}]{DBLP:journals/tplp/AngelisFPP15}
{\sc {De Angelis}, E.}, {\sc Fioravanti, F.}, {\sc Pettorossi, A.}, {\sc and}
  {\sc Proietti, M.} 2015.
\newblock Proving correctness of imperative programs by linearizing constrained
  {H}orn clauses.
\newblock {\em {TPLP}\/}~{\em 15,\/}~4-5, 635--650.

\bibitem[\protect\citeauthoryear{Dutertre}{Dutertre}{2014}]{Dutertre:cav2014}
{\sc Dutertre, B.} 2014.
\newblock Yices 2.2.
\newblock In {\em {CAV}}, {A.~Biere} {and} {R.~Bloem}, Eds. {LNCS}, vol. 8559.
  Springer, 737--744.

\bibitem[\protect\citeauthoryear{Esparza, Kiefer, and Luttenberger}{Esparza
  et~al\mbox{.}}{2007}]{DBLP:conf/stacs/EsparzaKL07}
{\sc Esparza, J.}, {\sc Kiefer, S.}, {\sc and} {\sc Luttenberger, M.} 2007.
\newblock On fixed point equations over commutative semirings.
\newblock In {\em {STACS}, Proceedings}. LNCS, vol. 4393. Springer, 296--307.

\bibitem[\protect\citeauthoryear{Esparza, Kiefer, and Luttenberger}{Esparza
  et~al\mbox{.}}{2010}]{DBLP:journals/jacm/EsparzaKL10}
{\sc Esparza, J.}, {\sc Kiefer, S.}, {\sc and} {\sc Luttenberger, M.} 2010.
\newblock Newtonian program analysis.
\newblock {\em J. {ACM}\/}~{\em 57,\/}~6, 33.

\bibitem[\protect\citeauthoryear{Esparza, Luttenberger, and Schlund}{Esparza
  et~al\mbox{.}}{2014}]{DBLP:conf/lata/EsparzaLS14}
{\sc Esparza, J.}, {\sc Luttenberger, M.}, {\sc and} {\sc Schlund, M.} 2014.
\newblock A brief history of strahler numbers.
\newblock In {\em {LATA}. Proceedings}, {A.~H. Dediu},
  {C.~Mart{\'{\i}}n{-}Vide}, {J.~L. Sierra{-}Rodr{\'{\i}}guez}, {and}
  {B.~Truthe}, Eds. {LNCS}, vol. 8370. Springer, 1--13.

\bibitem[\protect\citeauthoryear{Gallagher}{Gallagher}{1993}]{gallagher:pepm93}
{\sc Gallagher, J.~P.} 1993.
\newblock Specialisation of logic programs: A tutorial.
\newblock In {\em Proceedings PEPM'93, ACM SIGPLAN Symposium on Partial
  Evaluation and Semantics-Based Program Manipulation}. ACM Press, Copenhagen,
  88--98.

\bibitem[\protect\citeauthoryear{Gallagher and Lafave}{Gallagher and
  Lafave}{1996}]{Gallagher-Lafave-Dagstuhl}
{\sc Gallagher, J.~P.} {\sc and} {\sc Lafave, L.} 1996.
\newblock Regular approximation of computation paths in logic and functional
  languages.
\newblock In {\em Partial Evaluation}, {O.~Danvy}, {R.~Gl{\"u}ck}, {and}
  {P.~Thiemann}, Eds. Springer-Verlag {LNCS}, vol. 1110. 115--136.

\bibitem[\protect\citeauthoryear{Ganty, Iosif, and Kone{\v{c}}n{\'{y}}}{Ganty
  et~al\mbox{.}}{2016}]{Ganty2016}
{\sc Ganty, P.}, {\sc Iosif, R.}, {\sc and} {\sc Kone{\v{c}}n{\'{y}}, F.} 2016.
\newblock Underapproximation of procedure summaries for integer programs.
\newblock {\em STTT\/}~{\em 19,\/}~5 (apr), 565--584.

\bibitem[\protect\citeauthoryear{Gonnord and Halbwachs}{Gonnord and
  Halbwachs}{2006}]{DBLP:conf/sas/GonnordH06}
{\sc Gonnord, L.} {\sc and} {\sc Halbwachs, N.} 2006.
\newblock Combining widening and acceleration in linear relation analysis.
\newblock In {\em {SAS}}, {K.~Yi}, Ed. {LNCS}, vol. 4134. Springer, 144--160.

\bibitem[\protect\citeauthoryear{Grebenshchikov, Gupta, Lopes, Popeea, and
  Rybalchenko}{Grebenshchikov
  et~al\mbox{.}}{2012}]{DBLP:conf/tacas/GrebenshchikovGLPR12}
{\sc Grebenshchikov, S.}, {\sc Gupta, A.}, {\sc Lopes, N.~P.}, {\sc Popeea,
  C.}, {\sc and} {\sc Rybalchenko, A.} 2012.
\newblock {HSF(C)}: A software verifier based on {H}orn clauses - (competition
  contribution).
\newblock In {\em TACAS}, {C.~Flanagan} {and} {B.~K{\"o}nig}, Eds. LNCS, vol.
  7214. Springer, 549--551.

\bibitem[\protect\citeauthoryear{Grebenshchikov, Lopes, Popeea, and
  Rybalchenko}{Grebenshchikov
  et~al\mbox{.}}{2012}]{DBLP:conf/pldi/GrebenshchikovLPR12}
{\sc Grebenshchikov, S.}, {\sc Lopes, N.~P.}, {\sc Popeea, C.}, {\sc and} {\sc
  Rybalchenko, A.} 2012.
\newblock Synthesizing software verifiers from proof rules.
\newblock In {\em {PLDI}}, {J.~Vitek}, {H.~Lin}, {and} {F.~Tip}, Eds. {ACM},
  405--416.

\bibitem[\protect\citeauthoryear{Gurfinkel, Kahsai, and Navas}{Gurfinkel
  et~al\mbox{.}}{2015}]{DBLP:conf/tacas/GurfinkelKN15}
{\sc Gurfinkel, A.}, {\sc Kahsai, T.}, {\sc and} {\sc Navas, J.~A.} 2015.
\newblock Sea{H}orn: {A} framework for verifying {C} programs (competition
  contribution).
\newblock See \citeN{DBLP:conf/tacas/2015}, 447--450.

\bibitem[\protect\citeauthoryear{Heizmann, Hoenicke, and Podelski}{Heizmann
  et~al\mbox{.}}{2009}]{DBLP:conf/sas/HeizmannHP09}
{\sc Heizmann, M.}, {\sc Hoenicke, J.}, {\sc and} {\sc Podelski, A.} 2009.
\newblock Refinement of trace abstraction.
\newblock In {\em {SAS}}, {J.~Palsberg} {and} {Z.~Su}, Eds. LNCS, vol. 5673.
  Springer, 69--85.

\bibitem[\protect\citeauthoryear{Heizmann, Hoenicke, and Podelski}{Heizmann
  et~al\mbox{.}}{2013}]{DBLP:conf/cav/HeizmannHP13}
{\sc Heizmann, M.}, {\sc Hoenicke, J.}, {\sc and} {\sc Podelski, A.} 2013.
\newblock Software model checking for people who love automata.
\newblock See \citeN{DBLP:conf/cav/2013}, 36--52.

\bibitem[\protect\citeauthoryear{Hermenegildo, Bueno, Carro,
  L{\'{o}}pez{-}Garc{\'{\i}}a, Mera, Morales, and Puebla}{Hermenegildo
  et~al\mbox{.}}{2012}]{DBLP:journals/tplp/HermenegildoBCLMMP12}
{\sc Hermenegildo, M.~V.}, {\sc Bueno, F.}, {\sc Carro, M.}, {\sc
  L{\'{o}}pez{-}Garc{\'{\i}}a, P.}, {\sc Mera, E.}, {\sc Morales, J.~F.}, {\sc
  and} {\sc Puebla, G.} 2012.
\newblock An overview of ciao and its design philosophy.
\newblock {\em {TPLP}\/}~{\em 12,\/}~1-2, 219--252.

\bibitem[\protect\citeauthoryear{Hoder and Bj{\o}rner}{Hoder and
  Bj{\o}rner}{2012}]{DBLP:conf/sat/HoderB12}
{\sc Hoder, K.} {\sc and} {\sc Bj{\o}rner, N.} 2012.
\newblock Generalized property directed reachability.
\newblock In {\em {SAT}. Proceedings}, {A.~Cimatti} {and} {R.~Sebastiani}, Eds.
  {LNCS}, vol. 7317. Springer, 157--171.

\bibitem[\protect\citeauthoryear{Hojjat, Konecn{\'{y}}, Garnier, Iosif, Kuncak,
  and R{\"{u}}mmer}{Hojjat et~al\mbox{.}}{2012}]{DBLP:conf/fm/HojjatKGIKR12}
{\sc Hojjat, H.}, {\sc Konecn{\'{y}}, F.}, {\sc Garnier, F.}, {\sc Iosif, R.},
  {\sc Kuncak, V.}, {\sc and} {\sc R{\"{u}}mmer, P.} 2012.
\newblock A verification toolkit for numerical transition systems - tool paper.
\newblock In {\em {FM}}, {D.~Giannakopoulou} {and} {D.~M{\'{e}}ry}, Eds.
  {LNCS}, vol. 7436. Springer, 247--251.

\bibitem[\protect\citeauthoryear{Jaffar, Maher, Marriott, and Stuckey}{Jaffar
  et~al\mbox{.}}{1998}]{JMMS}
{\sc Jaffar, J.}, {\sc Maher, M.}, {\sc Marriott, K.}, {\sc and} {\sc Stuckey,
  P.} 1998.
\newblock The semantics of constraint logic programs.
\newblock {\em Journal of Logic Programming\/}~{\em 37,\/}~1--3, 1--46.

\bibitem[\protect\citeauthoryear{Jones, Gomard, and Sestoft}{Jones
  et~al\mbox{.}}{1993}]{Jones-Gomard-Sestoft}
{\sc Jones, N.}, {\sc Gomard, C.}, {\sc and} {\sc Sestoft, P.} 1993.
\newblock {\em {P}artial {E}valuation and {A}utomatic {S}oftware {G}eneration}.
\newblock Prentice Hall.

\bibitem[\protect\citeauthoryear{Kafle and Gallagher}{Kafle and
  Gallagher}{2017}]{kafleG2015Horn}
{\sc Kafle, B.} {\sc and} {\sc Gallagher, J.~P.} 2017.
\newblock Horn clause verification with convex polyhedral abstraction and tree
  automata-based refinement.
\newblock {\em Computer Languages, Systems {\&} Structures\/}~{\em 47}, 2--18.

\bibitem[\protect\citeauthoryear{Kafle, Gallagher, and Ganty}{Kafle
  et~al\mbox{.}}{2016}]{DBLP:journals/corr/KafleGG16}
{\sc Kafle, B.}, {\sc Gallagher, J.~P.}, {\sc and} {\sc Ganty, P.} 2016.
\newblock Solving non-linear horn clauses using a linear horn clause solver.
\newblock In {\em {HCVS}}, {J.~P. Gallagher} {and} {P.~R{\"{u}}mmer}, Eds.
  {EPTCS}, vol. 219. 33--48.

\bibitem[\protect\citeauthoryear{Kafle, Gallagher, and Morales}{Kafle
  et~al\mbox{.}}{2016}]{DBLP:conf/cav/KafleGM16}
{\sc Kafle, B.}, {\sc Gallagher, J.~P.}, {\sc and} {\sc Morales, J.~F.} 2016.
\newblock {RAHFT}: {A} {T}ool for {V}erifying {H}orn {C}lauses {U}sing
  {A}bstract {I}nterpretation and {F}inite {T}ree {A}utomata.
\newblock In {\em {CAV}}, {S.~Chaudhuri} {and} {A.~Farzan}, Eds. Lecture Notes
  in Computer Science, vol. 9779. Springer, 261--268.

\bibitem[\protect\citeauthoryear{McMillan and Rybalchenko}{McMillan and
  Rybalchenko}{2013}]{McmillanR2013}
{\sc McMillan, K.~L.} {\sc and} {\sc Rybalchenko, A.} 2013.
\newblock Solving constrained {H}orn clauses using interpolation.
\newblock Tech. rep., Microsoft Research.

\bibitem[\protect\citeauthoryear{Mordvinov and Fedyukovich}{Mordvinov and
  Fedyukovich}{2017}]{LPAR-21:Synchronizing_Constrained_Horn_Clauses}
{\sc Mordvinov, D.} {\sc and} {\sc Fedyukovich, G.} 2017.
\newblock Synchronizing constrained horn clauses.
\newblock In {\em LPAR-21. 21st International Conference on Logic for
  Programming, Artificial Intelligence and Reasoning}, {T.~Eiter} {and}
  {D.~Sands}, Eds. EPiC Series in Computing, vol.~46. EasyChair, 338--355.

\bibitem[\protect\citeauthoryear{Nielson and Nielson}{Nielson and
  Nielson}{1992}]{DBLP:books/daglib/0067731}
{\sc Nielson, H.~R.} {\sc and} {\sc Nielson, F.} 1992.
\newblock {\em Semantics with applications - a formal introduction}.
\newblock Wiley professional computing. Wiley.

\bibitem[\protect\citeauthoryear{Peralta, Gallagher, and Sa\u{g}lam}{Peralta
  et~al\mbox{.}}{1998}]{Peralta-Gallagher-Saglam-SAS98}
{\sc Peralta, J.}, {\sc Gallagher, J.~P.}, {\sc and} {\sc Sa\u{g}lam, H.} 1998.
\newblock Analysis of imperative programs through analysis of constraint logic
  programs.
\newblock In {\em {SAS}}, {G.~Levi}, Ed. Springer-Verlag {LNCS}, vol. 1503.
  246--261.

\bibitem[\protect\citeauthoryear{Reps, Turetsky, and Prabhu}{Reps
  et~al\mbox{.}}{2016}]{DBLP:conf/popl/RepsTP16}
{\sc Reps, T.~W.}, {\sc Turetsky, E.}, {\sc and} {\sc Prabhu, P.} 2016.
\newblock Newtonian program analysis via tensor product.
\newblock In {\em {POPL}}, {R.~Bod{\'{\i}}k} {and} {R.~Majumdar}, Eds. {ACM},
  663--677.

\bibitem[\protect\citeauthoryear{R{\"u}mmer, Hojjat, and Kuncak}{R{\"u}mmer
  et~al\mbox{.}}{2013}]{DBLP:conf/cav/RummerHK13}
{\sc R{\"u}mmer, P.}, {\sc Hojjat, H.}, {\sc and} {\sc Kuncak, V.} 2013.
\newblock Disjunctive interpolants for {H}orn-clause verification.
\newblock See \citeN{DBLP:conf/cav/2013}, 347--363.

\bibitem[\protect\citeauthoryear{Sharygina and Veith}{Sharygina and
  Veith}{2013}]{DBLP:conf/cav/2013}
{\sc Sharygina, N.} {\sc and} {\sc Veith, H.}, Eds. 2013.
\newblock {\em {CAV}}. {LNCS}, vol. 8044. Springer.

\end{thebibliography}

\end{document}